\documentclass[letterpaper, 10 pt, conference]{ieeeconf} 
\IEEEoverridecommandlockouts

\usepackage{enumitem}
\setlist{topsep=0pt, leftmargin=*}
\usepackage{algorithm}
\usepackage{algorithmicx}

\usepackage[noend]{algpseudocode}
\usepackage{amsmath, amsfonts, bm, amssymb, tabularx}
\usepackage{latexsym, mathtools}
\usepackage{amsthm}
\usepackage{ifthen}
\usepackage{hyperref}\hypersetup{colorlinks=true, unicode=true, linkcolor=[rgb]{0.10,0.05,0.67}, citecolor=[rgb]{0.10,0.05,0.67}, filecolor=[rgb]{0.10,0.05,0.67}, urlcolor=[rgb]{0.10,0.05,0.67}}
\usepackage{Shorthands}
\usepackage{cleveref}
\newcommand{\citep}[1]{\cite{#1}}
\usepackage{wrapfig}
\usepackage{etoolbox}
\AtBeginEnvironment{algorithm}{\setlength{\intextsep}{6pt} \setlength{\textfloatsep}{6pt}}

\usepackage[numbers, sort&compress]{natbib}
\allowdisplaybreaks

\title{\LARGE \bf
 Policy Gradient for LQR with Domain Randomization
}

\author{
Tesshu Fujinami$^{1}$, Bruce D. Lee, Nikolai Matni, George J. Pappas
\thanks{$^1$ All authors are with the Department of Electrical and Systems Engineering, University of Pennsylvania. Emails: \tt\small\{ftesshu, brucele, nmatni, pappasg\}@seas.upenn.edu}%
}

\begin{document}

\twocolumn

\maketitle
\thispagestyle{empty}
\pagestyle{empty}

\begin{abstract}
Domain randomization (DR) enables sim-to-real transfer by training controllers on a distribution of simulated environments, with the goal of achieving robust performance in the real world. Although DR is widely used in practice and is often solved using simple policy gradient (PG) methods, understanding of its theoretical guarantees remains limited. 
Toward addressing this gap, we provide the first convergence analysis of PG methods for domain-randomized linear quadratic regulation (LQR). We show that PG converges globally to the minimizer of a finite-sample approximation of the DR objective under suitable bounds on the heterogeneity of the sampled systems. We also quantify the sample-complexity associated with achieving a small performance gap between the sample-average and population-level objectives.
Additionally, we propose and analyze a discount-factor annealing algorithm that obviates the need for an initial jointly stabilizing controller, which may be challenging to find. Empirical results support our theoretical findings and highlight promising directions for future work, including risk-sensitive DR formulations and stochastic PG algorithms.
\end{abstract}
\section{Introduction}

Domain randomization (DR) has emerged as a dominant paradigm to enable transfer of policies optimized in simulation to the real world by randomizing simulator parameters during training \citep{tobin2017dr,peng2018drforcontrol, akkaya2019solving}. In doing so, just as with robust control, DR accounts for discrepancies between the model used in simulation to synthesize a policy and the system that it is deployed on. However, unlike conventional robust control approaches, DR minimizes an average control objective over the uncertainty in the system rather than a worst case objective. 
Since DR does not solely focus on optimizing the worst-case performance, it can result in less conservative controller performance while still ensuring robust stability with high probability. 
Furthermore, DR can be easily implemented via first order methods. 
This makes it straightforward to incorporate into a wide variety of reinforcement learning schemes and to benefit from the increasing availability of parallel computation.

Despite the ease with which DR can be implemented using first order methods, ensuring convergence of these methods remains a critical challenge, with practitioners relying upon complex scheduling of various hyperparameters in the optimization procedure \citep{akkaya2019solving}. Motivated by this challenge, we rigorously study the convergence of policy gradient methods for DR in the setting of the linear quadratic regulator.

\subsection{Related Work}
\paragraph{Domain Randomization}
Domain randomization, introduced by \citet{tobin2017dr}, is widely used for enabling \emph{sim-to-real transfer}. By randomizing simulator parameters during training, it aims to produce policies robust to simulator variations, thereby enabling transfer to the real-world. The sampling distribution may be a design variable in the synthesis problem \citep{mehta2020active}, or come from a learning procedure \citep{fujinami2025domainrandomizationsampleefficient}, e.g. as the posterior distribution of a Bayesian identification procedure. Recently, the approach has been applied in areas like autonomous racing \citep{loquercio2019deep}
and robotic control \citep{peng2018drforcontrol, akkaya2019solving}.
The control community has used similar randomized approaches to achieve high-probability guarantees (e.g. scenario optimization and related methods \citep{calafiore2006scenario, stengel1991technical, ray1993monte, vidyasagar2001randomized}).
 Recent work has explored generalization of domain randomization in discrete Markov Decision Processes \citep{chen2021understanding, zhong2019pacreinforcementlearningrealworld} and for continuous control \citep{fujinami2025domainrandomizationsampleefficient}. However, finding optimization procedures which minimize the domain randomization objective remains a challenge, with many empirical studies relying on heurisic approaches like curriculum learning \citep{muratore2022robot}. In this work, we rigorously study the convergence of policy gradient for minimizing the domain randomization objective for linear quadratic control. \textcolor{red}{}

\paragraph{Policy Gradient Methods}
Global convergence of policy gradient methods was first established for linear quadratic control by \citep{fazel2018global}, relying upon a condition known as \emph{gradient dominance} to demonstrate that driving the gradient to zero suffices to ensure convergence to the optimal solution. Subsequently, the results have been extended to a multitude of settings including mixed $\mathcal{H}_2/\mathcal{H}_{\infty}$ \citep{zhang2020policy}, Markovian jump systems \citep{rathod2021global}, and output estimation \citep{umenberger2022globally}. \citet{hu2023toward} provide a comprehensive survey on convergence studies of policy optimization methods for continuous control. There are two works closely related to domain randomization: \citet{gravell2019learning} and \citet{wang2023model}. \citet{gravell2019learning} study a related setting with multiplicative noise on the system, which resembles domain randomization but with variations drawn independently at each time step. 
Consequently, the resulting optimal controller exhibits robustness primarily in an average sense, rather than explicitly accounting for structured uncertainty as in domain randomization. \citet{wang2023model} study a federated approach to minimize the average LQR cost over a collection of systems; however, they prove convergence only up to a gap defined in terms of the heterogeneity between the systems. In contrast, we prove that as long as the systems are sufficiently similar, an approximate gradient domination condition suffices to ensure convergence of policy gradient methods to the globally optimal solution. We focus exclusively on model-based analysis; the extension of our results to the model free setting involves routine extension of the analysis for the variance of gradient estimates from prior work \citep{fazel2018global}.

\subsection{Contributions}

We study the convergence of policy gradient methods applied to the domain randomized LQR objective and establish the following results:
\begin{itemize}
    \item \textbf{DR Objective Approximation:} We propose optimizing a sample average approximation of the DR objective defined in terms of a finite number of sampled systems. We characterize the excess cost incurred by the solution to the approximate problem on the original objective. 
    \item \textbf{Policy Gradient Convergence:} We prove that a first order policy gradient algorithm converges to the globally optimal solution for the sample average approximation of the DR objective. This is the first result showing that policy gradient methods converge to the optimal solution for domain randomized LQR problems, removing the heterogeneity bias that appears in prior work \citep{wang2023model}. 
    \item \textbf{Joint Stabilization:} We propose a first order gradient descent algorithm based on discount factor annealing which provably converges to a controller stabilizing a collection of dynamical systems, extending prior work \citep{perdomo2021stabilizing} to the multi-system setting. 
\end{itemize}
The above contributions establish the convergence properties of first order gradient-based algorithms applied to the DR objective for continuous control. Given the prevalence of DR for the practice of learning-enabled robotic control and the engineering challenges involved in achieving convergent optimization procedures for the DR objective, we view this result as an important step towards the analysis and design of reliable reinforcement learning algorithms for sim-to-real transfer. Motivated by this, we include a number of empirical results highlighting limitations and open questions regarding the analysis of this  work. These include the generalization of the DR objective to risk metrics other than the expectation and the convergence of stochastic policy gradient algorithms. 
\section{Problem Formulation}
    We consider a fully observed linear time-invariant system defined by an unknown parameter $\theta \in\R^{d_{\theta}}$ with state $x_t \in \R^{\dx}$, control input $u_t \in \R^{\du}$, and disturbance $w_t \in \R^{\dx}$. The state of the system evolves as
    \begin{align*}
        x_{t+1} = A(\theta) x_t + B(\theta) u_t + w_t,
    \end{align*}
    where $w_t \overset{i.i.d.}{\sim} \mathcal{N}(0, I)$ is a zero-mean $i.i.d.$ disturbance.\footnote{Generalization to noise covariance $I  \neq \Sigma_w  \succ 0$ may be achieved by scaling $A(\theta) \gets \Sigma_w^{1/2} A(\theta)$ and $B(\theta) \gets \Sigma_w^{1/2} B(\theta)$.} For a fixed parameter $\theta$, 
    we define the infinite horizon linear quadratic regulation cost function in terms of positive definite matrices $Q \succeq I$ and $R = I$ as
    \begin{align*}
        J(K, \theta) \triangleq \limsup_{T\to\infty}\frac{1}{T}\E^{K} \brac{\sum_{t=0}^{\infty} x_t^\top Q x_t + u_t^\top R u_t \bigg\vert \theta},
    \end{align*}
    where the superscript $K$ indicates that the expectation is taken under the feedback policy $u_t = K x_t$.\footnote{Generalization to $Q \nsucceq I$ $R\neq I$ can be achieved by rescaling the cost by $1/\lambda_{\min}(Q)$ and the input as $B(\theta) \gets B(\theta)R^{-1/2}$.} The goal of DR is to determine a controller $K$ which minimizes the expectation of the LQR objective over a distribution of systems, as in prior work on sampling-based control \citep{vidyasagar2001randomized}. Such a distribution can be hand-designed to capture parametric uncertainty, or learned through Bayesian system identification \citep{fujinami2025domainrandomizationsampleefficient}. In particular, we let $\Theta$ be a random variable with density $p_{\Theta}$. The domain randomization problem is then to find $K_{DR}^\star \triangleq \argmin_{K} J_{DR}(K)$, with 
    \begin{align}
        \label{eq: DR objective}
        J_{DR}(K) \triangleq \E_{\Theta} J(K, \Theta). 
    \end{align}
    
    To solve this problem, we propose applying gradient descent to a finite sample approximation to objective \eqref{eq: DR objective}.
    We use $M$ independently sampled systems $\theta_1,\dots, \theta_M \sim p_{\Theta}$ to construct such a finite-sample approximation:
      \begin{align}
        \label{eq: sample average}
        J_{SA}(K) \triangleq \frac{1}{M} \sum_{i=1}^M J(K, \theta_i),
    \end{align}
    and apply policy gradient methods to find its minimizer. In particular, we start from a controller $K_0$, and iteratively update the controller as 
    \begin{align}
        \label{eq: grad update}
        K_{n+1} = K_n - \frac{\alpha }{M} \sum_{i=1}^M \nabla_K J(K_n, \theta_i),
    \end{align}
     where $\alpha$ is a fixed stepsize and, as shown in \citet{fazel2018global}, the gradient $\nabla_K J(K, \theta_i)$ for system $i$ is given by 
    \begin{align}
        \nabla_K J(K, \theta_i) &\triangleq E_K^i \Sigma_K^i. \label{eq: gradient of the cost function}
    \end{align}
    In the above expression, $\Sigma_K^i$ is the steady state covariance of system $i$ under controller $K$, and $E_K^i$ can be written in terms of the Lyapunov equation defining the LQR cost of applying controller $K$ to system $i$:
    \begin{align*}
        \Sigma_K^i &\triangleq \dlyap(A(\theta_i) + B(\theta_i) K, \Sigma_w) \\
        E_K^i &\triangleq 2 (R + B(\theta_i)^\top P_K^i B(\theta_i)) K + B(\theta_i)^\top P_K^i A(\theta_i) \\
        P_K^i &\triangleq \dlyap(A(\theta_i) + B(\theta_i) K, Q + K R K^\top).
    \end{align*}
    
    The primary technical challenge in proving convergence of the above scheme to the minimizer of \eqref{eq: DR objective} is that the approach to show gradient domination by \citet{fazel2018global} does not generalize to the sample average LQR cost for $M \geq 2$. Consequently, we introduce an additional assumption which ensures the support of the distribution $p_{\Theta}$ is suitably small.
    \begin{assumption}[Heterogeneity bound]
    \label{asmp: heterogeneity}
        Let $\calS \subseteq \R^{d_{\theta}}$ be such that $p_{\Theta}(\theta) = 0$ for $\theta \notin \calS$, and that for all $\theta_1, \theta_2 \in \calS$
        \begin{align*}
            \norm{\bmat{A(\theta_1) & B(\theta_1)} - \bmat{A(\theta_2) & B(\theta_2)}} \leq \varepsilon_{\mathsf{HET}},
        \end{align*}
        where 
        \begin{align*}
            \varepsilon_{\mathsf{HET}} \leq \frac{1}{5e5 \tau_B(\theta) \trace\left(P_{K(\theta)}^\theta\right)^{11/2}}
        \end{align*}
        for all $\theta \in \calS$ and where $\tau_B \triangleq \max\curly{1, \sup_{\theta \in \calS} \norm{B(\theta)}}$. Here $K(\theta)$ is the LQR controller for system $\theta$, and $P_{K(\theta)}^\theta$ its corresponding cost matrix from the Riccati equation.
    \end{assumption}
    The bound on $\varepsilon_{\mathsf{HET}}$ is a technical condition for our analysis. It scales inversely with $\norm{B(\theta)}$  and the trace of the Riccati equation. In particular, if the systems are too heterogeneous, our analysis cannot guarantee convergence. 
    

    Given the above condition, we demonstrate the convergence of policy gradient methods. The proof consists of three parts. We first show that policy gradient methods converge to the minimizer of the sample average objective $J_{SA}(K)$ starting from an initial controller $K_0$ which simultaneously stabilizes $\theta_1, \dots, \theta_M$ and achieves a sufficiently small cost $J_{SA}(K_0)$. We then introduce a discount factor scheduling scheme to achieve convergence from an arbitrary initial controller. Finally, we characterize the discrepancy between the empirical objective~\eqref{eq: sample average} and its population counterpart~\eqref{eq: DR objective}.  
\section{Policy Gradient Convergence}
\label{s: SA convergence}

We now prove convergence of the gradient update  \eqref{eq: grad update} to the minimizer of the sample average LQR objective \eqref{eq: sample average}. Denote this minimizer by $K_{SA}^\star \triangleq \argmin_K J_{SA}(K).$ To show convergence, we impose an additional assumption that the initial controller simultaneously stabilizes all systems.
\begin{assumption}[Simultaneous Stabilization]
    \label{asmp: simultaneous stabilization}
    We assume that $K_0\in\mathcal{K}$
    where 
    \[
        \mathcal{K} \coloneqq \curly{K \colon \rho(A(\theta_i) + B(\theta_i) K) < 1, ~ i=1,\dots, M}.
    \] 
\end{assumption}

To achieve sufficient conditions to ensure that \eqref{eq: grad update} converges starting from such a $K_0$, we introduce a definition that characterizes the cost of a controller on a particular sample in terms of the sample average cost of that controller.

\begin{definition}(Scenario Boundedness)
    \label{def: scenario boundedness}
    We call the collection of scenarios $\theta_1, \dots, \theta_M$ \emph{$(B,\slack)$-bounded} if for all $K$ satisfying $J_{SA}(K) \leq B J_{SA}(K_{SA}^\star)$, it holds that $J(K, \theta_i) \leq \slack J_{SA}(K)$. 
\end{definition}

Any collection of scenarios is $(B,\slack)$-bounded with $B =\infty$ and $\slack=M$ by the fact that the individual cost can never exceed the average scaled by the number of scenarios. However, if the scenarios are sufficiently close, we can determine a $(B,\slack)$-stability condition such that $\slack$ does not scale with the number of scenarios. This corresponds to the case where one single controller can perform fairly well in multiple scenarios. One such sufficient condition is provided in \Cref{lem: scenario boundedness}.

Under the above definition, we show that the sample average objective is coercive with a smoothness bound defined in terms of the average cost of the initial iterate $K_0$ on the systems $\theta_1, \dots, \theta_M$. This will in turn be used to ensure convergence of \eqref{eq: grad update} to a fixed point. 
Note that for any $\theta_i$, the LQR cost $J(K, \theta_i)$ is coercive, twice continuously differentiable, and $L_i$ smooth over any sub-level set $\mathcal{K}_{\slack\zeta}^i \triangleq \curly{K\in\mathcal{K} \colon J(K, \theta_i)\leq \slack\zeta}$ for $\zeta < \infty$~\citep[Lemma 1]{hu2023toward} for $L_i$ a polynomial in system parameters. This enables us to establish the following lemma. 



\begin{lemma}[Sum of LQR Costs is Coercive]
    \label{lem: coercivity}
    Suppose that the collection $(\theta_1, \dots, \theta_M)$ is $(B,\slack)$-bounded. Let $\mathcal{K}_\zeta \triangleq \curly{K\in\mathcal{K} \colon J_{SA}(K)\leq\zeta}$. 
    It holds that $\mathcal{K}_\zeta \subseteq \cap_{i=1}^M\mathcal{K}_{\slack\zeta}^i$ for any $\zeta \leq B$. 
    Furthermore, it follows that $J_{SA}(K)$ is coercive, twice continuously differentiable, and $L$-smooth over $\mathcal{K}_\zeta$ where $L = \frac{\sum_{i=1}^ML_i}{M}$.
\end{lemma}
\begin{proof}
    By \Cref{def: scenario boundedness}, it holds that for any $i \in \curly{1,\dots, M}$, the cost is bounded as $J(K,\theta_i) \leq \slack J_{SA}(K) \leq \slack \zeta$ for all $K \in \calK_{\zeta}$. Consequently, $\mathcal{K}_\zeta \subseteq \cap_{i=1}^M\mathcal{K}_{\slack\zeta}^i$. Now, since $MJ_{SA}(K) \geq J(K, \theta^i)$ for $i=1, \cdots, M$, when one of $J(K,\theta^i)\rightarrow\infty$ for any sequence $\curly{K^l}_{l=1}^{\infty}$ if $\norm{K^l}_2\rightarrow\infty$ or $K^l$ converges to an element of $\partial \calK$, $J_{SA}(K)\rightarrow\infty$ also holds. Thus $J_{SA}(K)$ is coercive~\citep[Definition 1]{hu2023toward}. 
    By the fact that $J(K, \theta_i)$ is twice differentiable and $L_i$ smooth over $\calK_{\slack \zeta}^i$, it holds by the mean value theorem that
    \begin{align}
        &J(K', \theta_i) - J(K, \theta_i) \nonumber\\
        &\leq (K'-K)^\top\nabla_KJ(K, \theta_i) + \frac{L_i}{2}\norm{K'-K}^2_F, \label{eq: Li-smoothness}
    \end{align}
    for $K, K'\in\mathcal{K}_{\slack\zeta}^i$. 
    Then $J_{SA}(K)$ is also twice continuously differentiable, and by summing \eqref{eq: Li-smoothness} over $i =1, \cdots, M$, it holds that $J_{SA}(K)$ is $L$-smooth on $\mathcal{K}_\zeta$ with $L = \frac{\sum_{i=1}^ML_i}{M}$.
\end{proof}

The smoothness constant $L_i$ can be explicitly written in terms of system parameters~\citep[Theorem 7]{fazel2018global}.
Instantiating these bounds, we find that the smoothness constant $L$ of $J_{SA}(K)$ over the sublevel set of value $J_{SA}(K_0)$  satisfies
\begin{align*}
    L \leq &\mathsf{poly}(\max_{i\in\curly{1, \cdots, M}}J(K_0, \theta_i), \max_{i\in\curly{1, \cdots, M}}\norm{\nabla J(K_0, \theta_i)}, \\
    &\max_{i\in\curly{1, \cdots, M}}\norm{A(\theta_i)}, \max_{i\in\curly{1, \cdots, M}}\norm{B(\theta_i)})
\end{align*}
A consequence of the above result is that the gradient descent procedure converges to a fixed point. 

\begin{lemma}[Convergence to a Fixed Point]
    \label{lem: fixed point for gradient}
    Suppose that the collection $(\theta_1, \dots, \theta_M)$ is $(B,\slack)$-bounded. Let $\zeta_0 = J_{SA}(K_0)$ and $L$ be the smoothness constant of $J_{SA}(K)$ on $\mathcal{K}_{\zeta_0}$.
    Consider the policy gradient method \eqref{eq: grad update}.  Then, for any $0 < \alpha < \frac{2}{L}$, we have $K_n\in\mathcal{K}_{\zeta_0}$ and $J_{SA}(K_{n+1}) \leq J_{SA}(K_n)$ for all n.
    Furthermore, we have the convergence of the gradient $\sum_{i=1}^M\nabla_K J_{SA}(K, \theta_i) \rightarrow 0$ with the rate
    \begin{align*}
        \min_{0\leq l\leq k}\norm{\nabla_K J_{SA}(K)}^2_F \leq \frac{\zeta_0}{C(k+1)},
    \end{align*}
    with $C = \alpha - \frac{L\alpha^2}{2} > 0$.
\end{lemma}

This is the direct consequence of \Cref{lem: coercivity} and Theorem 1 of~\cite{hu2023toward}. 
\Cref{lem: fixed point for gradient} ensures convergence of the policy gradient procedure to a fixed point; however, it does not imply that this fixed point is the global minimizer of \eqref{eq: sample average}. In the case of a single dynamical system, \citet{fazel2018global} show that a property known as gradient domination is satisfied by the LQR cost: \[J(K, \theta) - J(K^\star, \theta) \leq \norm{\Sigma_{K^\star}}\norm{\nabla_K J(K, \theta)}_F^2.\] It is not clear if such a condition holds for the sum of LQR objectives defined by multiple systems. In particular, extending the analysis in Lemma 11 of \citep{fazel2018global} leads to the following weaker result, which provides only an approximate gradient domination condition.
\begin{lemma}[Approximate Gradient Domination]
    For any $K$ that simultaneously stabilizes $\theta_1, \dots, \theta_M$, it holds that 
    \begin{align*}
        J_{SA}(\!K)\! -
        \!J_{SA}(K_{SA}^\star)\!\leq\! \norm{\frac{1}{M}\!\sum_{i=1}^M\! \nabla_K J(\!K, \theta_i) (\Sigma_K^i)^{\!-\!1} \Sigma_{K_{SA}^\star}^i}_F^2\!\!\!.
    \end{align*}
\end{lemma}
\begin{proof}
    Let $\widetilde{E} = \sum_{i=1}^M\paren{\Sigma_{K_{SA}^\star}^i \kron I}\VEC E_K^i$ and $\widetilde{G} = \sum_{i=1}^M\Sigma_{K_{SA}^\star}^i \kron \paren{B(\theta_i)^TP_K^iB(\theta_i) + R}$. By Lemma 12 of \citep{fazel2018global} and the identity $\trace(A^\top BCD^\top) = \VEC(A)^\top (D\kron B)\VEC (C)$, we have that 
    \begin{align*}
        &J_{SA}(K)- J_{SA}(K_{SA}^\star) = \frac{1}{M}\sum_{i=1}^M J(K, \theta_i) - J(K_{SA}^\star, \theta_i) \\
        &= -\frac{1}{M}\VEC\paren{K - K_{SA}^\star}^\top(\widetilde{G} \VEC (K - K_{SA}^\star) -2\widetilde{E}) \\
        &\leq \frac{1}{M}\widetilde{E}\widetilde{G}^{-1}\widetilde{E} \leq \frac{1}{M^2} \norm{\widetilde{E}}^2 = \norm{\frac{1}{M}\sum_{i=1}^ME_K^i\Sigma_{K_{SA}^\star}^i}^2_F,
    \end{align*}
    where the first inequality results from completing the square and the second inequality follows from the fact that $\tilde G \succeq M I$. The last equality holds by the identity $\VEC(XYZ) = (Z^\top \otimes X) \VEC(Y)$. Conclude by applying \eqref{eq: gradient of the cost function}. 
\end{proof}

Critically, we observe that convergence of the gradient $\frac{1}{M} \sum_{i=1}^M \nabla_K J(K, \theta_i)$ to zero does not immediately suffice to ensure that the sum of gradients weighted by $(\Sigma_K^i)^{-1} \Sigma_{K_{SA}^\star}^i$ appearing in the approximate gradient domination condition converges to zero. However, we may write the upper bound in the condition as a gradient term, plus a remainder as
\begin{equation}
\begin{aligned}
    \label{eq: grad dom plus remainder}
    &J_{SA}(K) \!-\! J_{SA}(K_{SA}^\star) \\
    \quad\quad&\!\leq\! 2\paren{\norm{\nabla_K J_{SA}(K)}_F^2 + \norm{\mathsf{R}(K, K_{SA}^\star)}_F^2\!},
\end{aligned}
\end{equation}
where
\[
    \mathsf{R}(K, K_{SA}^\star) \triangleq \frac{1}{M} \sum_{i=1}^M E_K^i\paren{\Sigma_{K_{SA}^\star}^i - \Sigma_K^i}.
\]
The remainder term above becomes small as $K \to K_{SA}^\star$. In particular, the remainder term is bounded by a small constant multiplied by the cost gap, we achieve gradient domination.
\begin{lemma}
    Let $\mu \in (0,\frac{1}{2})$ and $\zeta_0 = J_{SA}(K_0)$. Suppose that 
    \begin{align}
        \label{eq: remainder bound}
        \norm{R(K, K_{SA}^\star)}_F^2 \leq \mu (J_{SA}(K) - J_{SA}(K_{SA}^\star))
    \end{align}
    for all $K \in S \subseteq \calK_{\zeta_0}$. Then for all $K \in S,$
    \begin{align*}
        J_{SA}(K) - J_{SA}(K_{SA}^\star) \leq \frac{2}{1 - 2 \mu} \norm{\nabla_K J_{SA}(K)}_F^2.
    \end{align*}
\end{lemma}
\begin{proof}
    The result follows  by substitution of \eqref{eq: remainder bound} into \eqref{eq: grad dom plus remainder}. 
\end{proof}
The above result in combination with \Cref{lem: fixed point for gradient} establishes convergence of algorithm \eqref{eq: grad update} to the solution of objective \eqref{eq: sample average} from an initial iterate satisfying Assumption~\ref{asmp: simultaneous stabilization} as long as the collection $\theta_1, \dots, \theta_M$ is $(B,\slack)$-bounded, and the iterates of the gradient updates \eqref{eq: grad update} converge to a set $S$ for which the condition \eqref{eq: remainder bound} is satisfied. The convergence rate is determined by the stepsize $\alpha$, the inital cost $\zeta_0$, and the smoothness constant $L$ of the sample average control cost. 

The heterogeneity bound of Assumption~\ref{asmp: heterogeneity} ensures that $(B, \slack)$-boundedness is satisfied. 
\begin{lemma}
    \label{lem: scenario boundedness}
    Under Assumption~\ref{asmp: heterogeneity}, it holds that any collection of samples $\theta_1, \dots, \theta_M$ drawn from $p_{\Theta}$ is $(B, \slack)$-bounded with $B = 8$ and $\slack = 2$.
\end{lemma}
If additionally the cost of the initial iterate is not too much higher than the optimal cost, then condition \eqref{eq: remainder bound} holds.

\begin{lemma}
    \label{lem: remainder bound}
    Suppose Assumption~\ref{asmp: heterogeneity} holds and that  
    \begin{align*}
        \zeta_0 = J_{SA}(K_0) \leq 8 J_{SA}(K_{SA}^\star). 
    \end{align*}
    Then condition \eqref{eq: remainder bound} is satisfied with $\mu = \frac{1}{4}$ and 
    \begin{align*}
        S = \curly{K \in \calK_{\zeta_0}: \norm{\nabla_K J_{SA}(K)}_F \leq \frac{1}{256 \tau_{B}  J_{SA}(K)^{5/2}}}.
    \end{align*}
\end{lemma}

Before proving this fact, we state an intermediate result, the proof of which is deferred to the appendix.
\begin{lemma}
    \label{lem: remainder by controller gap}
    Suppose that Assumption~\ref{asmp: heterogeneity} holds and that $J_{SA}(K_0) \leq 8 \inf_{K} J_{SA}(K)$. Then 
    \begin{align*}
        &\norm{R(K, K_{SA}^\star)}_F \leq \frac{1}{4} \norm{K-K_{SA}^\star}_F \\ &\norm{R(K_{SA}^\star, K)}_F \leq \frac{1}{4} \norm{K-K_{SA}^\star}_F
    \end{align*}
    for all $K \in S$ with $S$ defined by \Cref{lem: remainder bound}.
\end{lemma}

We now prove \Cref{lem: remainder bound}.
\begin{proof}
    By \Cref{lem: remainder by controller gap}, it holds that $\norm{R(K, K_{SA}^\star)}_F \leq \frac{1}{4} \norm{K-K_{SA}^\star}_F$ and $\norm{R( K_{SA}^\star, K)}_F \leq \frac{1}{4} \norm{K-K_{SA}^\star}_F$. To bound the result by $J_{SA}(K) - J_{SA}(K_{SA}^\star)$, note that by~\citep[Lemma 12]{fazel2018global},
\begin{align*}
    &J_{SA}(K) - J_{SA}(K_{SA}^\star) \\
    &\geq \frac{2}{M}\trace\paren{(K-K_{SA}^\star)^T\paren{\sum_{i=1}^ME_{K_{SA}^\star}^i\Sigma_{K}^i}} \!+\! \norm{K-K_{SA}^\star}_F^2 \\
    &\geq \norm{K-K_{SA}^\star}_F^2 - 2 \norm{K-K_{SA}^\star}_F \norm{R(K_{SA}^\star, K)}_F. 
\end{align*}
 It follows that $J_{SA}(K) - J_{SA}(K_{SA}^\star) \geq \frac{1}{2} \norm{K-K_{SA}^\star}_F^2$, which can in turn be substituted to show $\norm{R(K, K_{SA}^\star)}_F \leq \frac{1}{4} \norm{K-K^\star}_F \leq \frac{1}{2} \sqrt{J_{SA}(K) - J_{SA}(K_{SA}^\star)}$. Squaring completes the proof. 
\end{proof}

Ultimately, synthesizing these results provides the following convergence bound.
\begin{theorem}[Policy Gradient for Sample Average LQR]
    \label{thm: sample average convergence}
    Suppose Assumption~\ref{asmp: heterogeneity} holds. Let $K^\star_{SA}$ be an optimal policy for the sample average objective \eqref{eq: sample average}, and let $K_0$ be a controller which simultaneously stabilizes $\theta_1,\dots, \theta_M$ achieving average cost $J_{SA}(K_0) \leq 8 J_{SA}(K^\star_{SA})$. Let $L$ be the smoothness constant of $J_{SA}(K)$ over $K_{\zeta_0}$. Consider running algorithm \eqref{eq: grad update} with stepsize $\alpha = 1/L$. Then after 
    \[
        N \geq 2 L \zeta_0 \max\curly{1/\varepsilon, 64 \tau_B \zeta_0} 
    \]
    iterations, let $\tilde K = \argmin_{K \in \curly{K_0, \dots, K_N}} \norm{\nabla_K J_{SA}(K)}_F$. Then it holds that $J_{SA}(\tilde K) - J_{SA}(K^\star_{SA}) \leq \varepsilon$.  
\end{theorem}


    

The above result demonstrates that under a sufficiently small randomization distribution, the approximate gradient domination condition suffices to ensure convergence to the global optimal solution for the sample average LQR objective. The heterogeneity requirements of Assumption~\ref{asmp: heterogeneity} are stronger than desired, as they exclude particular cases in which it is possible to achieve a controller that simultaneously stabilizes multiple systems, but for which the systems are not close in the sense of Assumption~\ref{asmp: heterogeneity}. We leave relaxing this condition or determining if it is fundamental to future work. The condition $J_{SA}(K_0)\le 8J_{SA}(K_{SA}^\star)$ means $K_0$ is not too far in cost from the optimal – this can be guaranteed by the discounting scheme in the sequel.

\section{Stabilization}

\begin{figure*}
    \centering
    \includegraphics[width=0.32\linewidth]{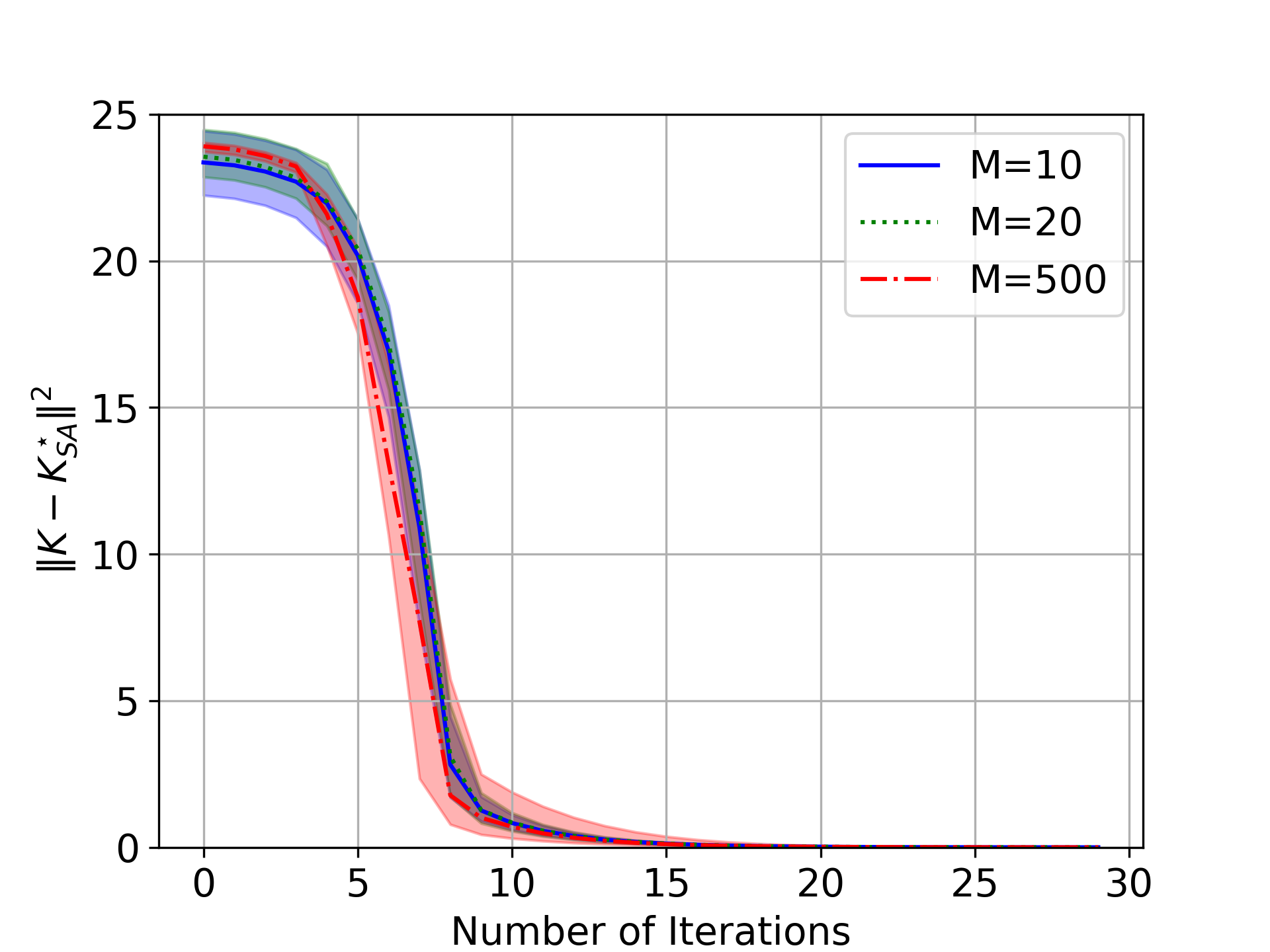}
    \includegraphics[width=0.32\linewidth]{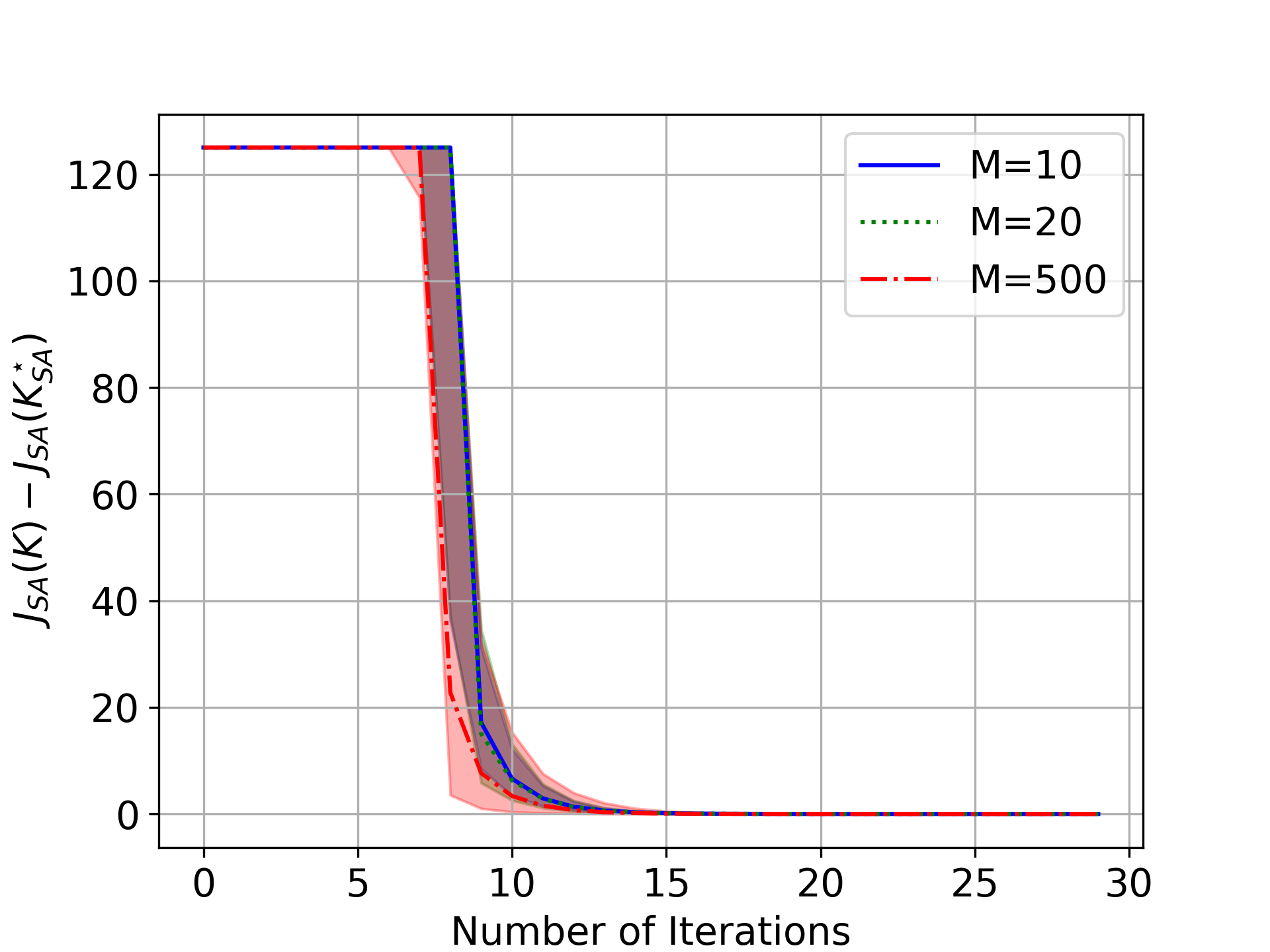}
    \includegraphics[width=0.32\linewidth]{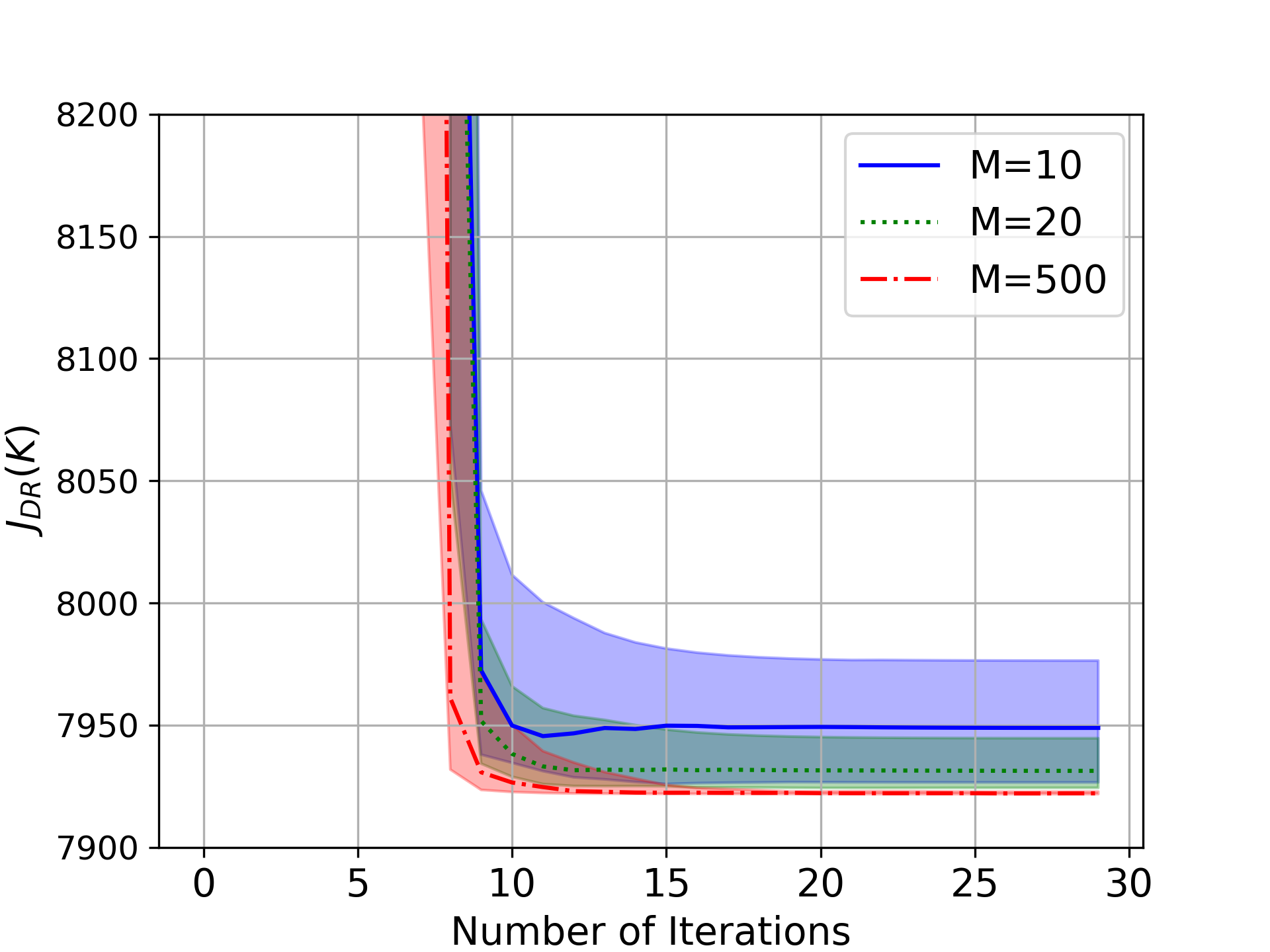}
    \caption{Convergence of policy gradient with domain randomization. 
    A controller initialized at $K=0$ converges to the optimal controller via progressive discounting and gradient descent (left), resulting in the convergence of the excess sample average cost (center). With more samples, the minimizer for the sample average cost and the domain randomization objective start to match (right). }
    \vspace{-16pt}
    \label{fig: verification}
\end{figure*}

The previous section analyzed the convergence of the policy gradient method on a collection of sampled instances starting from an initial stabilizing controller. Here, we remove the requirement that we have access to an initial controller which simultaneously stabilizes all systems and incurs a small cost.  We do so by incorporating a discounting factor annealing scheme. In particular, consider scaling the matrices $A(\theta_i)$ and $B(\theta_i)$ by a factor $\sqrt{\gamma}$. For any controller $K$, if $\gamma < \min_{i \in \curly{1,\dots,M}} \rho(A(\theta_i)+B(\theta_i)K)^{-2}$, then $K$ stabilizes $\sqrt{\gamma} A(\theta_i)$ and $\sqrt{\gamma} B(\theta_i)$ for $i=1,\dots,M$.\footnote{Scaling $A$ and $B$ by $\sqrt{\gamma} < \rho(A + BK)^{-1}$ is equivalent to introducing a discount factor in the LQR cost which ensures the cost is finite even for a controller which does not stabilize $A$, $B$.} Therefore, if we start from a small discount factor, run several iterations of \eqref{eq: grad update}, and then increase the discount factor, we will be able to ensure convergence to a controller that simultaneously stabilizes all systems. \Cref{alg: progressive discounting} extends a similar scheme proposed in~\citet{perdomo2021stabilizing} for the single system setting to the multi-system setting we consider. We let $J_{SA}(\cdot ~|~ \gamma)$ denote the sample average cost~\eqref{eq: sample average} evaluated on the discounted systems $(\sqrt{\gamma} A(\theta_i), \sqrt{\gamma} B(\theta_i))$. 

\begin{algorithm}
    \caption{Discount Annealing} 
    \label{alg: progressive discounting}
    \begin{algorithmic}[1]
    \State \textbf{Input:} collection of systems, $\theta_1, \dots, \theta_M$, optimization tolerance $\varepsilon$
    \State $K \gets 0$
    \State Find the largest $\gamma \in (0, 1]$ satisfying 
    \begin{align}
        \label{eq: initial gamma}
        J_{SA}(K \vert \gamma) \leq 8 \dx. 
    \end{align} 
    \While {$\gamma < 1$}
    \State Set $K \gets K'$, where $K'$ is  such that
    \begin{align}   
        \label{eq: controller update}
        J_{SA}(K' ~|~ \gamma) - \inf_{\tilde K} J_{SA}(\tilde K ~|~ \gamma) \leq \dx .
    \end{align}
    \State Find a discount factor $\gamma'\in[\gamma, 1]$ such that
    \begin{align}
        2.5J_{SA}(K ~|~ \gamma) \leq J_{SA}(K ~|~ \gamma') \leq 4J_{SA}(K ~|~ \gamma) \label{eq: gamma search}
    \end{align}
    \State $\gamma \gets \gamma'$
    \EndWhile
    \State Run \eqref{eq: grad update} starting from $K$ to find $K'$ such that \[J_{SA}(K') - \inf_{\tilde K} J_{SA}(\tilde K) \leq  \varepsilon.\]
    \State \textbf{Return} $K'$
    \end{algorithmic}
\end{algorithm}

Subproblems \eqref{eq: initial gamma} and \eqref{eq: gamma search} can be solved via bisection, while subproblem \eqref{eq: controller update} can be solved with algorithm \eqref{eq: grad update} starting from $K_0 \gets K$ by choosing an appropriate stepsize, and running sufficiently many iterations, as in \Cref{thm: sample average convergence}.
Extending the analysis of \citep[Theorem 1]{perdomo2021stabilizing} to the setting of multiple systems leads to the following result. 
\begin{theorem}[Stabilization by Discount Annealing]
    \label{thm: stabilization by progressive discounting}
    Suppose that Assumption~\ref{asmp: heterogeneity} holds. 
    Let $\gamma_0$ be the largest $\gamma$ satisfying \eqref{eq: initial gamma}. Then \Cref{alg: progressive discounting} is guaranteed to cease after at most $1024 J_{SA}(K^\star)^4\log(1/\gamma_0)$ iterations. At the start of any iteration, it holds that $J_{SA}(K\vert \gamma) \leq 8 \inf_{\tilde K} J_{SA}(K \vert \gamma)$. The final controller $K'$ jointly stabilizes all $M$ systems, and satisfies $J_{SA}(K') - J_{SA}(K^\star_{SA}) \leq \varepsilon$. 
\end{theorem}
Notably, the bound on the discounted cost of a controller $K$ at the start of every iteration suffices to satisfy the requisite condition on the cost of the $K_0$ in \Cref{thm: sample average convergence}. 

\begin{proof}
    We first show that at the start of every iteration $J_{SA}(K \vert \gamma) \leq 8 \inf_{\tilde K}  J_{SA}(\tilde K \vert \gamma)$. This holds for the first iteration by the fact that $\gamma_0$ satisfies the bound \eqref{eq: initial gamma}, and $\dx \leq \inf_{\tilde K} J_{SA}(\tilde K \vert \gamma)$. It holds for subsequent iterations by the fact that update \eqref{eq: controller update} achieves a controller within a factor of $2$ of the optimal cost, and the choice of $\gamma$ in \eqref{eq: gamma search} inflates the cost by at most a factor of $4$. This implies that $J_{SA}(K \vert \gamma) \leq 8 \inf_{\tilde K} J_{SA}(K\vert \gamma)$ by the fact that the cost is nondecreasing with $\gamma$. 

    To bound the number of iterations, let $P_{K,\gamma}^i$ denote the Ricatti equation solution corresponding to the discounted system  $(\sqrt(\gamma) A(\theta_i), \sqrt{\gamma} B(\theta_i))$. 
    By~\citet[Theorem 1]{perdomo2021stabilizing}, if $\tilde \gamma \leq \min_{i\in \curly{1,\dots, M}}\paren{1/(8 \norm{P_{K,\gamma}^i}^4)+1}^2 \gamma$, then $\trace(P_{K, \gamma}^i) \leq 2\trace(P_{K, \gamma}^i)$ for $i \in \curly{1,\dots, M}$. Therefore $J_{SA}(K \vert \tilde \gamma) \leq 2 J_{SA}(K \vert \gamma)$. Thus the value $\gamma'$ from \eqref{eq: gamma search} satisfies \[\frac{\gamma'}{\gamma} \! \geq \! \min_{i\in \curly{1,\dots, M}}\Bigg(\!\frac{1}{8 \norm{P_{K,\gamma}^i}^4}\!+\!1\!\Bigg)^2 \! \geq\! \paren{\!\frac{1}{2^{11} J_{SA}(K_{SA}^\star)^4}\! +\!1\!}^2. \]
    This follows from the bound $\norm{P_{K, \gamma}^i} \leq \trace(P_{K,\gamma}^i) \leq 2 J_{SA}(K \vert \gamma) \leq 4 J_{SA}(K^\star_{SA} \vert \gamma)$. Here,  the second inequality follows from \Cref{lem: scenario boundedness} paired with the fact that $J_{SA}(K \vert \gamma) \leq 8 \inf_{\tilde K} J_{SA}(\tilde K \vert \gamma)$. The last inequality follows from the fact that $K$ satisfies the error bound of \eqref{eq: controller update}, and $\dx \leq J_{SA}(K_{SA}^\star \vert \gamma)$. Then using the inequality $\log(1+x) \leq x$, it holds that $\gamma = 1$ after the given number of iterations. 

    The bound on the sample average cost of the final iterate follows immediately from the last step of the algorithm.
\end{proof}


\section{Sampling}

\label{s: sampling}

The previous sections provided guarantees ensuring that a gradient-based algorithm converges to the minimizer of objective \eqref{eq: sample average}. In this section, we bound the suboptimality incurred by optimizing the sample average cost~\eqref{eq: sample average} rather than the desired domain randomization objective~\eqref{eq: DR objective}. In particular, let $\tilde K$ be a the iterate of running the aforementioned gradient procedure. It holds that
\begin{align*}
    &J_{DR}(\tilde K) \!-\! J_{DR}(K_{DR}^\star) \!=\! J_{DR}(\tilde K) \!-\! J_{SA}(\tilde K) \\
    &+ J_{SA}(\tilde K) - J_{SA}(K_{DR}^\star) + J_{SA}(K_{DR}^\star) - J_{DR}(K_{DR}^\star). 
\end{align*}
The results of the previous section ensure that the term 
\begin{align*}
    J_{SA}(\tilde K) - J_{SA}(K_{DR}^\star) &\leq J_{SA}(\tilde K) - J_{SA}(K_{SA}^\star) \leq \varepsilon.
\end{align*}
Then to characterize the gap $J_{DR}(\tilde K) -J_{DR}(K^\star_{DR})$, it suffices to achieve upper bounds on $\abs{J_{DR}(\tilde K) \!-\! J_{SA}(\tilde K)}$ and $\abs{J_{DR}(K_{DR}^\star) - J_{SA}(K_{DR}^\star)}$. Note that if $J(K, \theta) \leq \calJ$, then by Hoeffding's inequality \citep{vershynin2009high}, it holds that with probability at least $1-\delta$, 
\begin{align*}
    \abs{J_{DR}(K) - J_{SA}(K)} \leq \sqrt{\frac{2 \calJ^2}{M} \log\frac{2}{\delta}}.
\end{align*} 
To establish the bound $\calJ$, consider $\tilde \theta, \theta \in \calS$. We show in \Cref{lem: uniform bound} that under Assumption~\ref{asmp: heterogeneity}, $J(K_{DR}^\star, \tilde \theta) \leq 6 J(K(\theta), \theta)$ and $J(\tilde K, \tilde \theta) \leq 48 J(K(\theta), \theta)$. This leads to the following characterization of the excess cost.

\begin{theorem}[Population Excess Cost Bound]
    \label{thm: population excess cost bound}
    Suppose Assumption~\ref{asmp: heterogeneity} holds.  Consider sampling $M$ systems from $\Theta$. Let $\tilde K$ be the iterate achieved by running \Cref{alg: progressive discounting}. Consider an arbitrary  $\theta \in \calS$. With probability at least $1-\delta$,
    \begin{align*}
        J_{DR}(\tilde K) - J_{DR}(K_{DR}^\star) \leq  \varepsilon + 200 \sqrt{\frac{J(K(\theta), \theta)}{M}\log\frac{2}{\delta}}.
    \end{align*}
\end{theorem}

\section{Numerical Experiments}
\label{s: numerical}

We conduct numerical experiments on the discretized and linearized inverted pendulum defined by
\begin{align}
    \label{eq: linearized pendulum}
    A = \bmat{1 & dt \\ \frac{g}{\ell} dt & 1}, 
    B= \bmat{0 \\ \frac{dt}{m \ell^2}}.
\end{align} 
We suppose that the parameters $dt = 0.01$ and $g=10$ are known. The unknown parameters $m$ and $\ell$ are modeled with a uniform distribution over the interval $[0.75, 1.25]$. For policy update in \Cref{alg: progressive discounting}, we run 20 iterations of \eqref{eq: grad update} with a stepsize $\alpha = 1\times10^{-3}$ instead of explicitly evaluating the inequality \eqref{eq: controller update}. Refer to the code for more details.\footnote{Code available at \href{https://github.com/Tesshuuuu/Policy-Gradient-for-LQR-with-Domain-Randomization}{this GitHub repository}}

\paragraph{Verification of Proposed Approach} \Cref{fig: verification} demonstrates the application of \Cref{alg: progressive discounting} to a collection of $M=10,20,500$ samples from the distribution of systems \eqref{eq: linearized pendulum}. We run the algorithm over $100$ random seeds. We plot the median and shade 25\% to 75\% quantile for the distance of the iterate from the optimal controller, the excess sample average control cost, and a Monte Carlo approximation of the domain randomized objective \eqref{eq: DR objective} formed with $10^5$ samples. 
We see that the discounted annealing algorithm succeeds in synthesizing a controller that simultaneously stabilizes all systems, and that the policy update \eqref{eq: grad update} converges to the minimizer of sample average cost. This is predicted by the analysis of  \Cref{thm: sample average convergence} and \Cref{thm: stabilization by progressive discounting}. The rightmost figure demonstrates the impact of optimizing the sample average objective defined in terms of $M$ samples. As suggested in \Cref{thm: population excess cost bound}, sampling more systems results in a controller achieving lower domain randomization objective. 

\paragraph{Alternative Risk Metrics} We can consider risk metrics other than the expectation for defining the objective \eqref{eq: DR objective}. For example, we could replace the expectation the entropic risk measure of temperature $t$: 
\begin{align}
    \label{eq: entropic risk}
    J_{ER}^t(K) \triangleq \frac{1}{t} \log \E_{\Theta} \brac{\exp(t J(K, \Theta))}.
\end{align}
The gradient-based optimization of a Monte Carlo approximation proposed in the previous sections can still be applied in this setting by replacing the average gradient of \eqref{eq: grad update} with the gradient of \eqref{eq: entropic risk}:
\begin{align*}
    K \gets K - \alpha \sum_{i=1}^M \frac{\exp(t J(K, \theta_i))}{\sum_{m=1}^M \exp(t J(K, \theta_i))} \nabla_K J(K, \theta_i).
\end{align*}
We lack theoretical convergence guarantees for this approach; however, \Cref{fig: entropic risk} validates it numerically. We see that the proposed algorithm still converges with this alternative risk metric and achieves the optimal controller for the risk-sensitive objective. This result opens the future direction of risk-sensitive domain randomized control. 
\begin{figure*}
    \centering
    \includegraphics[width=0.32\linewidth]{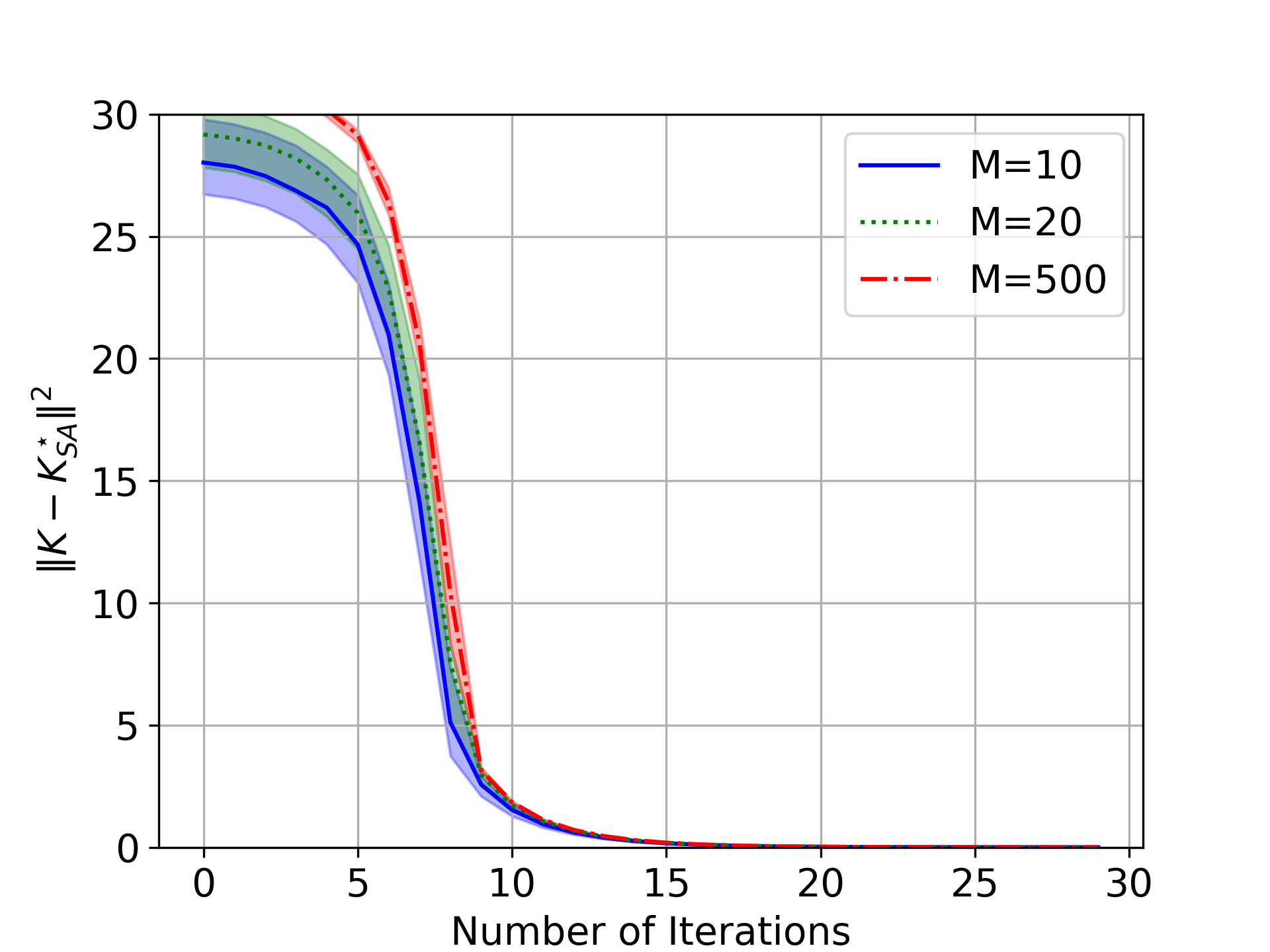}
    \includegraphics[width=0.32\linewidth]{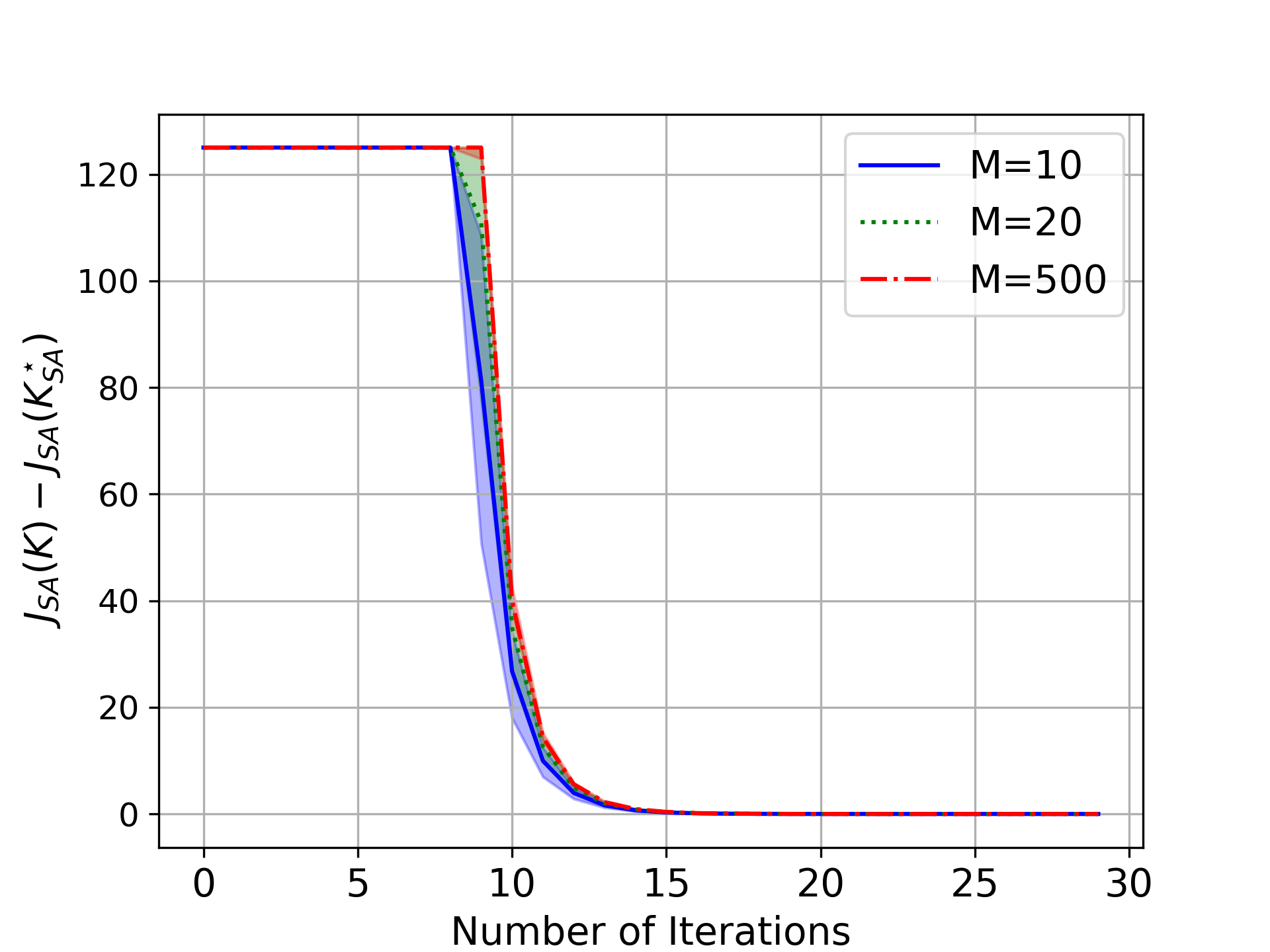}
    \includegraphics[width=0.32\linewidth]{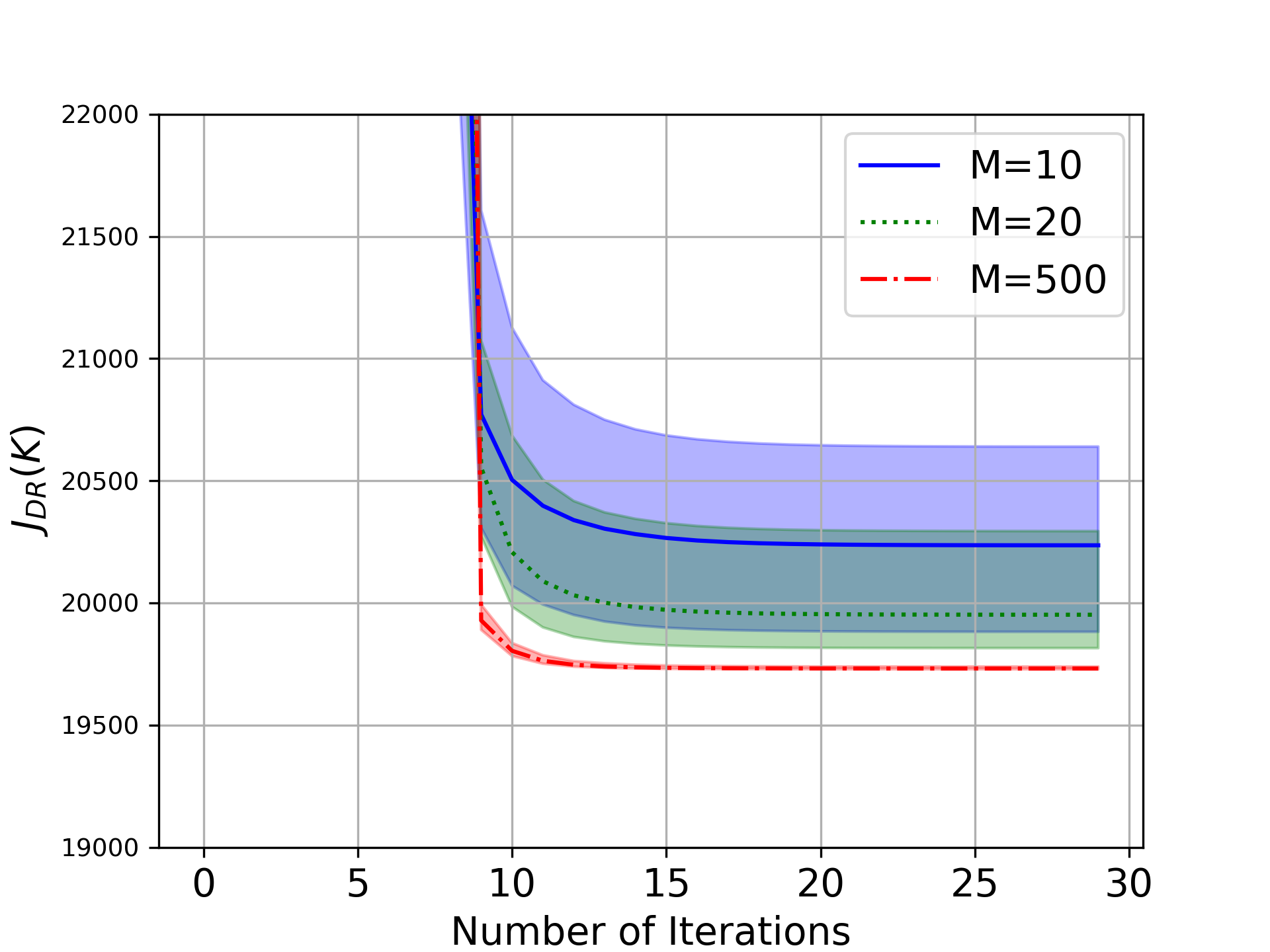}
    \vspace{-6pt}
    \caption{Convergence of policy gradient with the entropic risk measure using $t=1.0$ }
    \vspace{-16pt}
    \label{fig: entropic risk}
\end{figure*}

\paragraph{Stochastic Gradient Descent} Prior experiments verified that the approach analyzed in sections \Cref{s: SA convergence}-\ref{s: sampling} is valid for the numerical example of stabilizing a pendulum linearized about the upright position. In \Cref{fig:sgd}, we demonstrate a stochastic gradient descent approach which samples a new system at every iteration and takes a gradient step with that system, detailed in \Cref{alg: sgd}.\footnote{The update is performed with the gradient of the discounted cost for the sampled system. The  discounting ensures stabilization by the current iterate.}  This also converges, although we lack theoretical guarantees theoretical convergence guarantees. The figure plots the mean and standard deviation over 5 random seeds. The cost $J_{DR}(K)$ is approximated via Monte Carlo, as in the previous experiments. The initial stepsize is chosen as $\alpha = 2\times 10^{-4}$ and the discount coefficient is $\gamma_0 = 0.99$. We see that the controller converges to achieve approximately the same cost on the domain randomization objective as the batch algorithm with $M=500$ samples from \Cref{fig: verification}.

\begin{algorithm}
    \caption{Stochastic Gradient Descent} 
    \label{alg: sgd}
    \begin{algorithmic}[1]
    \State \textbf{Input:} distribution $p_{\Theta}$, iterations $N$, stepsize $\alpha$, discount coefficient $\gamma_0$
    \State $K \gets 0$, 
    \For{$n = 1, \dots,N$}
    \State Sample $\theta \sim p_{\Theta}$
    \State $\gamma \gets \min\curly{\gamma_0 \rho(A(\theta)+B(\theta) K)^{-2}, 1}$
    \State $K\gets K - \alpha \nabla_K J(K, \theta)$
    \EndFor
    \State \textbf{Return: } $K$
    \end{algorithmic}
\end{algorithm}

\begin{figure}
    \centering
    \includegraphics[width=0.7\linewidth]{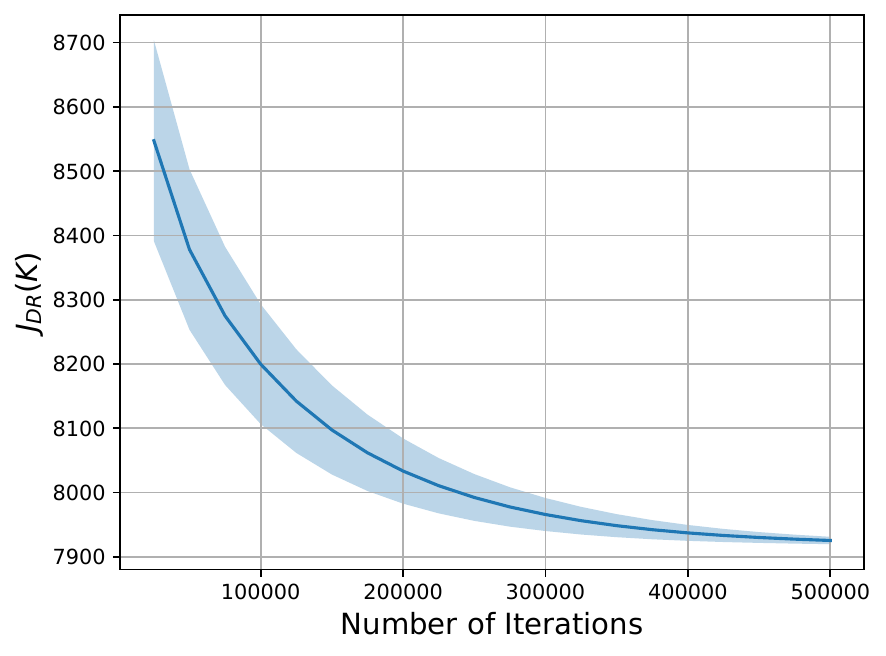}
    \vspace{-10 pt}
    \caption{Stochastic Gradient Descent (\Cref{alg: sgd}) applied to the linearized inverted pendulum of \eqref{eq: linearized pendulum}.}
    \vspace{-20pt}
    \label{fig:sgd}
\end{figure}


\paragraph{Rotary Inverted Pendulum Experiments}
We also verify the \Cref{alg: progressive discounting} on a Quanser Qube Servo 2 rotary inverted pendulum \citep{quanserqube} to study how it can help to overcome the sim-to-real gap. We consider the scenario where we have coarse estimates of several physical parameters, and compare the controller synthesized with DR against the LQR controller which takes these coarse estimates as ground truth. The DR controller achieves a higher success rate by holding a pendulum upright and around the center.\footnote{The experiment video can be found at \href{https://youtu.be/t6AbI1hNZg4}{this YouTube link}.} Details are deferred to the appendix. 
\section{Discussion}


There are several exciting possibilities for theoretical and empirical extensions of the results presented in this paper. 
\begin{itemize}[noitemsep, nolistsep]
    \item \textbf{Relaxed heterogeneity assumptions: } Our theory required strong requirements on the distance between systems. These requirements can almost certainly be made weaker; however, it is unclear the extent to which they are fundamental. Our numerical experiments failed to discover any instances where policy gradient did not converge due to a large distance between systems, unless it was impossible to simultaneously stabilize these systems. 
    \item \textbf{Extensions of the convergence analysis: } Our numerical experiments tested a variation of the algorithm (stochastic gradient descent) and an alternative risk metric (entropic risk) for domain randomization for which we do not have theoretical convergence guarantees, and saw empirical success. It would be interesting to provide convergence proofs for these approaches. In particular, stochastic gradient descent closer resembles the way that domain randomization is implemented in practice. Meanwhile, the entropic risk (or alternative metrics) may better represent the goal of domain randomization to provide robustness to uncertainty. Indeed, entropic risk interpolates between DR as $t\to 0$, and robust control as $t\to\infty$.
    \item \textbf{Practical implications: } The study conducted in this paper provides a better understanding of the application of gradient-based approaches for reinforcement learning with domain randomization. It may be possible to use this perspective to design alternative hyperparameter schedules for domain randomization as applied in practice \citep{akkaya2019solving}. 
\end{itemize}

\section{Conclusion}

Our work proves the convergence of policy gradient methods applied to linear quadratic control with domain randomization. The analysis relies upon a generalization of the gradient domination condition from prior work. This generalization results in guarantees that are sensitive to the initialization; however, a curriculum learning approach with an appropriate schedule can bypass this sensitivity. We believe that this line of analysis has potential to demystify and improve heuristic approaches for taming the optimization landscape of DR from the robot learning literature \citep{akkaya2019solving}.

\section*{Acknowledgements}
We thank Leonardo Toso for instructive conversations. 
TF is supported by JASSO Exchange Support program and UTokyo-TOYOTA Study Abroad Scholarship. TF and GP are supported in part by NSF Award SLES 2331880 and NSF TRIPODS EnCORE 2217033.  BL and NM are supported by NSF Award SLES-2331880, NSF CAREER award ECCS-2045834 and AFOSR Award FA9550-24-1-0102.

\bibliographystyle{IEEEtranN}
\bibliography{refs}

\onecolumn

\section{Appendix}
\label{s: policy gradient convergence}
\subsection{Notation}
We use the fact that $J(K, \theta)\geq 1$ by the lower bound on $Q$ and the fact that $\Sigma_w = I$. 
Note that $\norm{X+Y}^2\leq2\norm{X}^2 + 2\norm{Y}^2$.
Let $\nu(K) \coloneqq 1 + \norm{K}$ and $\zeta(K) \coloneqq 1 + \norm{K}^2$.  
Let $\tau_B = \max\curly{\norm{B(\theta_1)}, \cdots, \norm{
B(\theta_M)}, 1}$. 
Assume that if $J_{SA}(K)\leq\mathcal{B}$, then $J^i(K)\leq\slack J_{SA}(K)$.

\subsection{Helper Lemmas}
\begin{lemma}[Cost Bound, Lemma 13 of \citep{fazel2018global}, Lemma 6 of \citep{fujinami2025domainrandomizationsampleefficient}]
    \label{lem: Cost Bound}
    Let $K$ be arbitrary gain that stabilizes the system. Then following inequalities hold 
    \begin{enumerate}
        \item $\norm{K}^2 \leq \norm{P_K^i}\leq\frac{J(K, \theta_i)}{\smin(\Sigma_{w})} = J(K, \theta_i)$, 
        \item $\norm{\Sigma_K^i}\leq\frac{J(K, \theta_i)}{\smin(Q)} \leq J(K, \theta_i)$,
        \item $\norm{E_K^i}_F^2 \leq \frac{1}{\smin(\Sigma_w)^2}\trace\paren{\Sigma_K^i(E_K^i)^\top E_K^i\Sigma_K^i} = \norm{\nabla_K J(K, \theta_i)}^2_F$, \\
        \item $1 \leq \nu(K) \leq 2\sqrt{J(K, \theta_i)}$, ~ $1\leq\zeta(K)\leq2J(K, \theta_i)$.
    \end{enumerate}
\end{lemma}

\begin{lemma}[Closed Loop System Perturbation]
    \label{lem: closed loop system perturbation}
    Let $\cl^i\coloneqq A(\theta_i) + B(\theta_i)K$ for $i = 1, 2$ with some stabilizing controller $K$. 
    Define the heterogeneity parameter $\het \coloneqq \max\curly{\norm{A(\theta_1)-A(\theta_2)}, \norm{B(\theta_1)-B(\theta_2)}}$. 
    Then $\norm{\cl^1 - \cl^2} \leq \nu(K)\het.$
\end{lemma}
\begin{proof}
    From the definition of $\cl^i$, we have
    \begin{align*}
        \norm{\cl^1 - \cl^2} &= \norm{\paren{A(\theta_1) - A(\theta_2)} - \paren{B(\theta_1) - B(\theta_2)}K} \leq \norm{A(\theta_1) - A(\theta_2)} + \norm{B(\theta_1) - B(\theta_2)}\norm{K}.
    \end{align*}
\end{proof}

\begin{lemma}[State Correlation Matrix Perturbation]
    \label{lem: state correlation matrix perturbation}
    Let $K$ be a stabilizing controller. Then for $i, j\in\curly{1, \cdots, M}$, it holds that
        $\norm{\Sigma_K^i - \Sigma_K^j} \leq 2J(K, \theta_i)J(K, \theta_j)^{3/2}\paren{\nu(K)\het + \zeta(K)\het^2}.$
\end{lemma}
\begin{proof}
    From the Lyapunov perturbation argument in Lemma 3 in \citep{fujinami2025domainrandomizationsampleefficient} and \Cref{lem: closed loop system perturbation}, 
    \begin{align*}
    \norm{\Sigma_K^i - \Sigma_K^j} &\leq \frac{\norm{\Sigma_{K}^i}\norm{\Sigma_{K}^j}}{\smin(\Sigma_w)}\paren{2\norm{\cl^j}\norm{\cl^i - \cl^j} + \norm{\cl^i - \cl^j}^2} \\
    &\leq 2\norm{\Sigma_{K}^i}\norm{\Sigma_{K}^j}\paren{\norm{P_K^j}^{1/2} \het(1 + \norm{K}) + \het^2(1+\norm{K}^2)}, 
    \end{align*}
    where we use $\norm{\cl^j}\leq\norm{P_K^j}^{1/2}$ in the second inequality. Apply \Cref{lem: Cost Bound} to conclude.
\end{proof}

\begin{lemma}[Lyapunov Perturbation]
    \label{lem: Lyapunov Perturbation}
    Then it holds that
        $\norm{P_K^i - P_K^j} \leq 2J(K, \theta_i)J(K, \theta_j)^{3/2}\paren{\nu(K)\het + \zeta(K)\het^2}.$
\end{lemma}
\begin{proof}
    From Lemma 3 in \citep{fujinami2025domainrandomizationsampleefficient}, it follows for $i, j \in\curly{1, \cdots, M}$ that
    \begin{align*}
        \norm{P_K^i - P_K^j} &\leq \frac{2\norm{P_K^i}\norm{P_K^j}}{\smin(Q + K^TRK)}\paren{\norm{\cl^j}\het(1 + \norm{K}) + \het^2(1 + \norm{K}^2)},
    \end{align*} 
    and apply \Cref{lem: Cost Bound} and $\norm{\cl^j}\leq\norm{P_K^j}^{1/2}$. 
\end{proof}

\begin{lemma}[$E_K$ perturbation]
    \label{lem: EK perturbation}
    Let $K$ be a stabilizing controller. Then it holds for $i,j\in\curly{1, \cdots, M}$ that
    \begin{align*}
        \norm{E_K^i - E_K^j} &\leq \paren{J(K, \theta_i) + \tau_BJ(K, \theta_j) + 2\tau_BJ(K, \theta_i)J(K, \theta_j)^{3/2}}\nu(K)\het + 2\tau_BJ(K, \theta_i)J(K, \theta_j)^{3/2}\zeta(K)\het^2.
    \end{align*}
\end{lemma}
\begin{proof}
    \begin{align*}
        \norm{E_K^i - E_K^j} &=\norm{\paren{R + B(\theta_i)^\top P_K^iB(\theta_i)}K - B(\theta_i)^\top P_K^iA(\theta_i) - \paren{R + B(\theta_j)^\top P_K^jB(\theta_j)}K + B(\theta_j)^\top P_K^jA(\theta_j)} \\
        &= \norm{-B(\theta_i)^\top P_K^i\cl^i + B(\theta_j)^\top P_K^j\cl^j}  \\
        &= \norm{\paren{B(\theta_j)^\top P_K^j - B(\theta_i)^\top P_K^i}\cl^i + B(\theta_j)^\top P_K^j\paren{\cl^i - \cl^j}} \\
        &\leq \norm{B(\theta_j)^\top P_K^j - B(\theta_i)^\top P_K^i} + \norm{B(\theta_j)^\top P_K^j}\norm{\cl^i-\cl^j},
    \end{align*}
    where we add and subtract $B(\theta_j)^\top P_K^j\cl^i$ in the third equality. 
    As for the first term, we have
    \begin{align*}
        \norm{B(\theta_j)^\top P_K^j - B(\theta_i)^\top P_K^i} &= \norm{\paren{B(\theta_j) - B(\theta_i)}^\top P_K^i + B(\theta_j)^\top(P_K^j - P_K^i)} \leq \het\norm{P_K^i} + \norm{B(\theta_j)}\norm{P_K^i - P_K^j}.
    \end{align*}
    Thus with \Cref{lem: Cost Bound} and \Cref{lem: Lyapunov Perturbation}, 
    \begin{align*}
        \norm{B(\theta_j)^\top P_K^j - B(\theta_i)^\top P_K^i} &\leq J(K, \theta_i)\het + 2\tau_BJ(K, \theta_i)J(K, \theta_j)^{3/2}\paren{\het(1+\norm{K}) + \het^2(1+\norm{K}^2}.
    \end{align*}
   For the remaining term $
        \norm{(B^j)^\top P_K^j}\norm{\cl^i-\cl^j} \leq \tau_BJ(K, \theta_j)\het(1 + \norm{K})$,
    from \Cref{lem: Cost Bound} and \Cref{lem: closed loop system perturbation}. 
    Since $\nu(K) > 1$, we get the result. 
\end{proof}

\begin{lemma}[Gradient Perturbation]
    \label{lem: gradient perturbation}
    Suppose that $J(K, \theta_i) \leq \slack J_{SA}(K)$ for $i = 1,\dots, M$. Let $\nabla J(K, \theta_i)$ be the gradient of the cost function of $i$th system with the controller $K$. It holds for $i,j\in\curly{1, \cdots, M}$ that
    \begin{align*}
        \sum_{j=1, j\neq i}^{M}\norm{\nabla_K J(K, \theta_i) - \nabla_K J(K, \theta_j)}_F \leq \paren{M_i^1 + M_i^3\norm{\nabla_KJ(K, \theta_i)}}\nu(K)\het + \paren{M_i^2+ M_i^3\norm{\nabla_KJ(K, \theta_i)}}\zeta(K)\het^2.
    \end{align*}
    where
    \begin{align*}
        M_i^1 &= MJ_{SA}(K)J(K, \theta_i) + M\tau_B\slack J_{SA}(K)^2+ 2M\tau_B\slack^{3/2} J_{SA}(K)^{5/2} J(K, \theta_i), \\
        M_i^2 &= 2M\tau_B\slack^{3/2} J_{SA}(K)^{5/2}J(K, \theta_i), \\
        M_i^3 &= 2M \slack^{1/2} J_{SA}(K)^{3/2}.
    \end{align*}
\end{lemma}
\begin{proof}
    Since $\nabla_K J(K, \theta_i) = E_K^i\Sigma_K^i$, 
    \begin{align*}
        \norm{\nabla_K J(K, \theta_i) - \nabla_K J(K, \theta_j)}_F &= \norm{E_K^i\Sigma_K^i - E_K^j\Sigma_K^j}_F  \leq \norm{E_K^i\Sigma_K^i\paren{(\Sigma_K^i)^{-1}(\Sigma_K^i - \Sigma_K^j)}}_F + \norm{\paren{E_K^i - E_K^j}\Sigma_K^j}_F \\
        &\leq \norm{\nabla_K J(K, \theta_i)}_F\norm{(\Sigma_K^i)^{-1}(\Sigma_K^i - \Sigma_K^j)} + \norm{E_K^i -E_K^j}\norm{\Sigma_K^j}_F.
    \end{align*}
    By applying \Cref{lem: Cost Bound}, \Cref{lem: EK perturbation} and \eqref{eq: SigmaK_inv_SigmaKstar - SigmaK}, we get
    \begin{align*}
        &S_{ij} = \norm{\nabla_K J(K, \theta_i) - \nabla_K J(K, \theta_j)}_F \leq \paren{M_{ij}^1 + M_{ij}^3\norm{\nabla_K J(K, \theta_i)}}\nu(K)\het + \paren{M_{ij}^2 + M_{ij}^3\norm{\nabla_K J(K, \theta_i)}}\zeta(K)\het^2,
    \end{align*}
    where
        $M_{ij}^1 = \paren{J(K, \theta_i) + \tau_BJ(K, \theta_j) + 2\tau_BJ(K, \theta_i)J(K, \theta_j)^{3/2}}J(K, \theta_j), \, 
        M_{ij}^2 = 2\tau_BJ(K, \theta_i)J(K, \theta_j)^{5/2}, \, M_{ij}^3 = 2J(K, \theta_j)^{3/2}.$
    Then by summing over $j \in \curly{1,\dots,M}\setminus\curly{i}$ and applying the Cauchy-Schwarz inequality concludes the proof. 
\end{proof}

\begin{lemma}[Controller Gap]
    \label{lem: controller gap}
    Let $K_i, K_i^\star$ be arbitrary and optimal controllers for the system with the cost $J(K, \theta_i)$ for $i \in\curly{1, \cdots, M}$. Then it holds that
        $\norm{K_i - K_i^\star}_F \leq \sqrt{J(K_i^\star, \theta_i)}\norm{\nabla J(K_i, \theta_i)}_F$.
\end{lemma}
\begin{proof}
    From Lemma 12 in \citep{fazel2018global}, 
    \begin{align*}
        J(K, \theta_i) - J(K_i^\star, \theta_i) 
        &= \VEC(K_i-K_i^\star)^\top\paren{\Sigma_{K_i}^i \kron (R + B(\theta_i)^\top P^i_{K_i^\star}B(\theta_i))}\VEC(K_i-K_i^\star),
    \end{align*}
    where the term linear in $K-K_i^\star$ disappears since $E^i_{K_i^\star} = (R + B(\theta_i)^\top P_{K_i^\star}^iB(\theta_i))K^\star - B(\theta_i)^\top P^i_{K_i^\star}A(\theta_i) = 0$. 
    Since $\Sigma_K^i \kron (R + B(\theta_i)^\top P^i_{K^\star}B(\theta_i)) \succeq I$, 
         It holds that $\norm{K_i-K_i^\star}_F^2 \leq J(K_i, \theta_i) - J(K_i^\star, \theta_i).$
    Now consider the gradient domination condition for $J(K, \theta_i)$. 
    From Lemma 11 in \citep{fazel2018global}, 
    $J(K_i, \theta_i) - J(K_i^\star, \theta_i) \leq \norm{\Sigma_{K_i^\star}^i}\norm{\nabla J(K_i, \theta_i)}^2_F.$
    By applying \Cref{lem: Cost Bound} we get the result.
\end{proof}

\begin{lemma}[Cost bound with hetergeneity]
    \label{lem: cost bound with heterogeneity}
    Suppose Assumption \ref{asmp: heterogeneity} holds. Consider any $\theta, \tilde \theta \in \calS$, and let $K(\theta), K(\tilde \theta)$ be the corresponding LQR controller. The following inequalities are satisfied:
    \begin{enumerate}
        \item $J(K(\theta), \tilde \theta) \leq 2 J(K(\tilde\theta), \tilde \theta)$.
        \item $J(K(\theta), \tilde \theta) \leq 3 J(K(\theta), \theta)$.
    \end{enumerate}
\end{lemma}
\begin{proof}
    By the fact that $\het \leq \frac{1}{54 \norm{P_{K(\theta)}^\theta}^5}$, the first fact follows from \citet{simchowitz2020naive}, Theorem 5 part two. Similarly, the second fact follows from the first fact in addition to \citet{simchowitz2020naive}, Theorem 5 part four. 
    
\end{proof}

\subsection{Proof of \Cref{lem: scenario boundedness}}
Consider any $K$ such that $J_{SA}(K) \leq 8 J_{SA}(K^\star_{SA})$. For some $\tilde \theta \in S$, it holds that $J(K, \tilde \theta) \leq J_{SA}(K) \leq 8 J_{SA}(K^\star_{SA}).$ 

By \Cref{lem: cost bound with heterogeneity}, $J(K(\theta), \tilde \theta) \leq 3J(K(\theta), \theta)$ for $\theta, \tilde \theta \in \calS$. Therefore any controller $K$ satisfying $J_{SA}(K) \leq 8 J_{SA}(K^\star_{SA})$ also satisfies $J(K, \tilde\theta) \leq 24J_{SA}(K(\theta), \theta)$.

Furthermore, from \Cref{lem: Lyapunov Perturbation}, it holds for any $j\in\curly{1, \cdots, M}$ and some $\tilde\theta$ that
\begin{align*}
    \trace\paren{P_K^{\theta_j} - P_K^{\tilde\theta}} \leq 2\trace\paren{P_K^{\theta_j}}\norm{P_K^{\tilde\theta}}\paren{\norm{P_K^{\tilde\theta}}^{1/2}(1 + \norm{K})\het + (1 + \norm{K}^2)\het^2} 
    \leq 8\trace\paren{P_K^{\theta_j}}\norm{P_K^{\tilde\theta}}^2\het.
\end{align*}
From \Cref{lem: cost bound with heterogeneity}, it follows that $\het\leq\frac{1}{16\cdot24^2 \cdot\trace\paren{P_{K(\theta)}^\theta}^5}\leq 
\frac{1}{16\cdot 24^2 \trace\paren{P_{K(\theta)}^\theta}^2}$.
Then by substituting the above bound on $J_{SA}(K(\theta), \theta)$, $\het$ satisfies $\het\leq\frac{1}{16\norm{P_K^{\tilde\theta}}^2}$, resulting in $J(K, \theta_j) \leq 2J(K, \tilde\theta) \leq 2J_{SA}(K)$.

\subsection{Uniform bounds}

\begin{lemma}
    \label{lem: uniform bound}
    Let $\theta \in \calS$. Under Assumption~\ref{asmp: heterogeneity}, it holds that for any $\tilde \theta \in \calS$
    \begin{itemize}
        \item $J(K_{DR}^\star, \tilde \theta) \leq 6 J(K(\theta), \theta)$.
        \item $J(\tilde K, \tilde \theta) \leq 48 J(K(\theta), \theta)$.
    \end{itemize}
\end{lemma}
\begin{proof}
    For the first fact , note that there exists some $\theta'$ such that $J(K_{DR}^\star,\theta') \leq J_{DR}(K_{DR}^\star) \leq J_{DR}(K(\theta)) \leq 3J(K(\theta), \theta)$ by \Cref{lem: cost bound with heterogeneity}. Then by the same Lyapunov perturbation argument as in the proof of \Cref{lem: scenario boundedness}, it follows that $J(K_{DR}^\star, \tilde \theta) \leq 6 J(K(\theta), \theta)$. For the second fact, we follow the same argument, but with the bound $J_{SA}(\tilde K) \leq 8 J_{SA}(K_{SA}^\star)$.
\end{proof}

\subsection{Proof of \Cref{lem: remainder by controller gap}}
\begin{proof}
    It holds by \Cref{lem: scenario boundedness} that the collection $\theta_1, \dots, \theta_M$ is $(B, \slack)$-bounded with $B = 8$ and $\slack = 2$. By the assumption that $J_{SA}(K) \leq 8 J_{SA}(K^\star_{SA})$, it holds that $J(K,\theta_i) \leq \slack J_{SA}(K)$ for $i \in\curly{1,\dots, M}$. 

    The remainder term is bounded as follows:
    \begin{align}
        \norm{\mathsf{R}(K, K_{SA}^\star)}_F &= \frac{1}{M} \sum_{i=1}^M E_K^i\Sigma_K^i\paren{(\Sigma_K^i)^{-1}(\Sigma_{K_{SA}^\star}^i - \Sigma_K^i)} 
        \leq \frac{1}{M} \sum_{i=1}^M\norm{\nabla_KJ(K, \theta_i)}_F\norm{(\Sigma_K^i)^{-1}(\Sigma_{K_{SA}^\star}^i - \Sigma_K^i)}, \label{eq:R}
    \end{align}
    Since it holds from the proof in Lemma 3 in \citep{fujinami2025domainrandomizationsampleefficient} that 
    \begin{align*}
        \Sigma_{K_{SA}^\star}^i - \Sigma_K^i &\preceq \Sigma_K^i\norm{\Sigma_{K_{SA}^\star}^i}\paren{2\norm{A(\theta_i)-B(\theta_i)K}\norm{B(\theta_i)}\norm{K-K_{SA}^\star} + \norm{B(\theta_i)}^2\norm{K - K_{SA}^\star}^2},
    \end{align*}
    we get that
    \begin{align}
        \norm{(\Sigma_K^i)^{-1}(\Sigma_{K_{SA}^\star}^i - \Sigma_K^i)}& = \sqrt{\norm{(\Sigma_K^i)^{-1}(\Sigma_{K_{SA}^\star}^i - \Sigma_K^j)^2(\Sigma_K^i)^{-1}}^2} \nonumber  \\
        &\leq \norm{\Sigma_{K_{SA}^\star}^i}\paren{2\norm{A(\theta_i)-B(\theta_i)K}\norm{B(\theta_i)}\norm{K-K_{SA}^\star} + \norm{B(\theta_i)}^2\norm{K - K_{SA}^\star}^2}, \label{eq: SigmaK_inv_SigmaKstar - SigmaK}
    \end{align}
    for $i\in\curly{1, \cdots, M}$.
    To simplify this reduction and substitute it into \eqref{eq: SigmaK_inv_SigmaKstar - SigmaK}, we must show that $\norm{K-K_{SA}^\star}$ becomes small. To this end, note that if $K_i^\star$ is the optimal policy for system $i$, then,
\begin{equation}
\begin{aligned}
    \norm{K-K_{SA}^\star} &= \norm{K - K_{SA}^\star + \sum_{i=1}^M\frac{1}{M}(K_i^\star - K_i^\star)} \leq \frac{1}{M}\sum_{i=1}^M\paren{\norm{K - K_i^\star} + \norm{K_{SA}^\star - K_i^\star}} \\
    &\leq \frac{1}{M}\sum_{i=1}^M\sqrt{J(K_i^\star, \theta_i)}\paren{\norm{\nabla_K J(K, \theta_i)}_F + \norm{\nabla_K J(K_{SA}^\star, \theta_i)}_F} \quad \mbox{(\Cref{lem: controller gap})}\\
    &\leq \frac{\sqrt{\slack J_{SA}(K)}}{M}\sum_{i=1}^M\paren{\norm{\nabla_K J(K, \theta_i)}_F + \norm{\nabla_K J(K_{SA}^\star, \theta_i)}_F}. \quad (\sqrt{J(K_i^\star, \theta_i)} \leq \sqrt{J(K, \theta_i)} \leq \sqrt{\slack J_{SA}(K)})
    \label{eq:K-Kstar mid}
\end{aligned}
\end{equation}
We will show that under the provided heterogeneity conditions, the above quantity becomes small as $\norm{\nabla_K J_{SA}(K)} \to 0$. 
In particular,  we have
\begin{align*}
    \norm{\nabla_K J(K, \theta_i)}_F &= \norm{\nabla_K J(K, \theta_i) + \frac{1}{M}\sum_{j=1, j\neq i}^M\paren{\nabla_K J(K, \theta_j) - \nabla_K J(K, \theta_j)}}_F \\
    &= \norm{\frac{1}{M} \sum_{i=1}^M\nabla_K J(K, \theta_i) + \frac{1}{M}\sum_{j=1, j\neq i}^M\paren{\nabla_K J(K, \theta_i) - \nabla_K J(K, \theta_j)}}_F \\
    &\leq \frac{1}{M}\paren{M\norm{\nabla_KJ_{SA}(K)}_F + \sum_{j=1,j\neq i}^M\norm{\nabla_K J(K, \theta_i) - \nabla_K J(K, \theta_j)}_F} \\
    &\leq \frac{1}{M}\paren{M \grad + \paren{M_i^1 + M_i^3\norm{\nabla_KJ(K, \theta_i)}}\nu(K)\het + \paren{M_i^2+ M_i^3\norm{\nabla_KJ(K, \theta_i)}}\zeta(K)\het^2}.
\end{align*}
Note where $\grad \triangleq \norm{\nabla_K J_{SA}(K)}_F$. 
 and the second term is characterized by a quadratic function of the heterogeneity with the coefficients defined in \Cref{lem: gradient perturbation}.
By grouping the term, we get
\begin{align*}
    \norm{\nabla J(K, \theta_i)}_F \leq \frac{M \grad + M_i^1\nu(K)\het + M_i^2\zeta(K)\het^2}{M - M_i^3(\nu(K)\het + \zeta(K)\het^2)} \leq \frac{2}{M}\paren{M \grad + M_i^1\nu(K)\het + M_i^2\zeta(K)\het^2} \triangleq  \Lambda^i(K, \grad, \het),
\end{align*}
where $\Lambda^i(K, \grad, \het)$ is the upper bound on the gradient. The second inequality in the above bound follows from the fact that for $\het \leq \left(8\slack^{3/2} J_{SA}(K)^{3/2}\right)^{-1}$ (demonstrated in \eqref{eq: heterogeneity derivation}), 
\begin{align}
    &M_i^3(\nu(K)\het + \zeta(K)\het^2) = 2M\slack J_{SA}(K)(\nu(K)\het + \zeta(K)\het^2) < \frac{M}{2}. 
\end{align}

Define $\Lambda(K, \grad, \het) \triangleq \frac{1}{M} \sum_{i=1}^M \Lambda^i(\grad, \het)$. It holds by definition of $M_i^1$ and $M_i^2$ along with \Cref{lem: Cost Bound} that 
\begin{align*}
     \Lambda(K, \grad, \het) &= \grad + \paren{(1+\slack \tau_B)J_{SA}(K)^2 + 2 \tau_B \slack^{3/2} J_{SA}(K)^{7/2}} \nu(K) \het + 2 \tau_B \slack^{3/2} J_{SA}(K)^{7/2} \xi(K) \het^2\\
    &\leq \grad + 8 \slack^{3/2} J_{SA}(K)^{4} \het + 4 \tau_B \slack^{3/2} J_{SA}(K)^{9/2} \het^2.
\end{align*} 
Return now to the remainder term. From \eqref{eq:R}, \eqref{eq: SigmaK_inv_SigmaKstar - SigmaK}, and \Cref{lem: Cost Bound}, 
\begin{align*}
    \norm{\mathsf{R}(K, K_{SA}^\star)}_F &\leq \frac{1}{M} \sum_{i=1}^M\norm{\nabla_K J(K, \theta_i)}J(K, \theta_i)^{3/2}\paren{2\tau_B\norm{K-K_{SA}^\star} + \tau_B^2\norm{K-K_{SA}^\star}^2} \\
    &\leq \slack^{3/2} J_{SA}(K)^{3/2} \Lambda(K,\grad,\het) \paren{2\tau_B\norm{K-K_{SA}^\star} + \tau_B^2\norm{K-K_{SA}^\star}^2} \quad \mbox{(gradient bound \& $J(K,\theta_i) \leq \slack J_{SA}(K)$)} \\
    &\leq 2\tau_B\slack^{3/2} J_{SA}(K)^{3/2}\Lambda(K,\grad,\het)\paren{1 + \tau_B \sqrt{\slack J_{SA}(K)} \Lambda(K,\grad,\het)}\norm{K-K_{SA}^\star}. \quad \mbox{(\eqref{eq:K-Kstar mid} \& gradient bound)} 
\end{align*}
For $K \in S$ and $\het \leq \left(5\cdot2^8 \tau_B J_{SA}(K)^{11/2}\right)^{-1},$ we reach the desired conclusion. It is shown below that Assumption~\ref{asmp: heterogeneity} along with the assumption $J_{SA}(K_0) \leq 8\inf_K J_{SA}(K)$ suffice to satisfy this condition on $\het$.   
\begin{equation}
    \label{eq: heterogeneity derivation}
    \begin{aligned}
    \left(5\cdot2^8 \tau_B J_{SA}(K)^{11/2}\right)^{-1} &\geq \left(5\cdot2^{11} \tau_B J_{SA}(K_{SA}^\star)^{11/2}\right)^{-1}
    & \geq \left(5e5 \tau_B J(K^\star_i, \theta_i)^{11/2}\right)^{-1},
    \end{aligned}
\end{equation}
where the final inequality follows by \Cref{lem: cost bound with heterogeneity}.

Since $J_{SA}(K_{SA}^\star)\leq J_{SA}(K)$ and $\Lambda(K_{SA}^\star, \grad, \het))\leq\Lambda(K, \grad, \het)$ hold, the argument follows in similar way and we get
\begin{align*}
    \norm{\mathsf{R}(K_{SA}^\star, K)}_F \leq \frac{1}{4}\norm{K - K_{SA}^\star},
\end{align*}
under the same heterogeneity condition. 

\end{proof}

\subsection{Sim-to-real Experiments Details}
We conduct experiments on a nonlinear, underactuated rotational inverted pendulum defined by
\begin{align*}
    \dot{x} = f(x, u), ~ x = \bmat{\theta ~ \alpha ~ \dot\theta ~ \dot\alpha}^\top\in\R^4, u\in\R^1,
\end{align*}
where $\theta$ is the angle of arm from the forward position, $\alpha$ is the angle of pendulum from upright, $\dot\theta$ and $\dot\alpha$ are angular velocities of respective angles, and the control input $u$ is a motor voltage. 
We linearize this system around the zero state $x_0 = \bmat{0; 0; 0; 0}$ and zero control input $u_0 = 0$ with $A_c = \frac{\partial f}{\partial x}|_{(x_0, u_0)}$ and $B_c = \frac{\partial f}{\partial u}|_{(x_0, u_0)}$ and discretize the system with $dt = 0.01$, resulting in the following discrete time linear system
\begin{align*}
    x_{t+1} = Ax_t + Bu_t, ~ A\in\R^{4\times 4}, B\in\R^{4\times 1}.
\end{align*}
We design a linear state feedback policy $u=Kx$ with $K\in\R^{1\times 4}$ to regulate this system to the origin. The LQR cost matrices used to determine this controller are set as $Q = I_4$ and $R = 5I_2$.

We consider the case where the mass $m$ and the length $l$ of the pendulum are uncertain. The ground truth values are $m = 0.024$ and $l = 0.129$, and we consider coarse estimates $\hat{m} = 0.05$ and $\hat{l} = 0.2$.  For DR, we set the distribution for both parameters as a uniform distribution: $m\sim [0.02, 0.08]$ and $l\sim[0.1, 0.3]$. For policy update in \Cref{alg: progressive discounting}, we schedule the step size and number of gradient descent iterations based on $\gamma$. While this is not strictly necessary (choosing a very small stepsize and running sufficiently many iterations suffices), scheduling these parameters results in convergence with far fewer total iterations. In particular the stepsize and number of iterations of gradient descent are chosen as follows: 
\begin{itemize}
    \item For $\gamma < 0.85$, $\alpha = 1e-3$ and $n_\mathsf{steps}=160$
    \item For $\gamma \in (0.85, 1.0)$, $\alpha=1e-5$ and $n_\mathsf{steps}=480$
    \item For $\gamma = 1.0$, $\alpha=2e-5$ and $n_\mathsf{steps}=400$
\end{itemize}
Refer to the \texttt{dr\_stabilization\_qube.py} for more details. From \Cref{fig:quanser gd}, we see that the algorithm succeeds in finding a stabilizing controller and converging to the minimizer. 


\begin{wrapfigure}{r}{0.45\linewidth}
    \centering
    \vspace{-10pt} 
    \includegraphics[width=\linewidth]{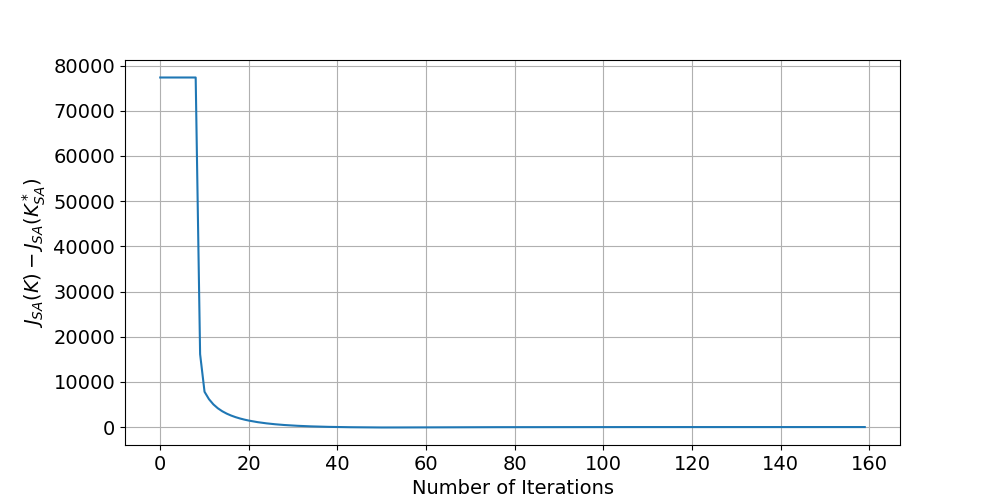}
    \caption{Gradient Descent (\Cref{alg: progressive discounting}) applied to the linearized rotational inverted pendulum.}
    \label{fig:quanser gd}
    \vspace{-10pt}
\end{wrapfigure}

We compare DR with certainty equivalence (CE), which synthesizes a policy by using the coarse parameter estimates as though they were the ground truth. The CE policy is calculated by solving discrete algebraic Ricatti equation. The resulting controllers are as follows:
\begin{align*}
    &K_{CE} = \bmat{-0.37 ~ 27.45 ~ -0.51 ~ 2.44}, \\
    &K_{DR} = \bmat{-0.41 ~ 31.62 ~ -0.66 ~ 2.97}.
\end{align*}
Since the system is linearized around the upright position, the linear policies determined via CE and DR  are used only when the pendulum is near upright, i.e. when the magnitude of $\alpha$ is less than 20 degrees, and an energy-based controller is used for swing up behavior. Refer to \texttt{control.py} for details of the energy-based controller.

In the experiment, we run the 10 trials from the same downward initial position, and measure the success rate. We define the success and failure as follows:
\begin{itemize}
    \item \textbf{Success}: the pendulum holds the upright posture and remains within $\pm 90$ degrees from the center for $10$ seconds
    \item \textbf{Upright}: the pendulum holds upright posture but fails to remain around center
    \item \textbf{Failure}: The pendulum fails to hold in the upright position
\end{itemize}

\begin{figure*}
    \centering
    \includegraphics[width=0.48\linewidth]{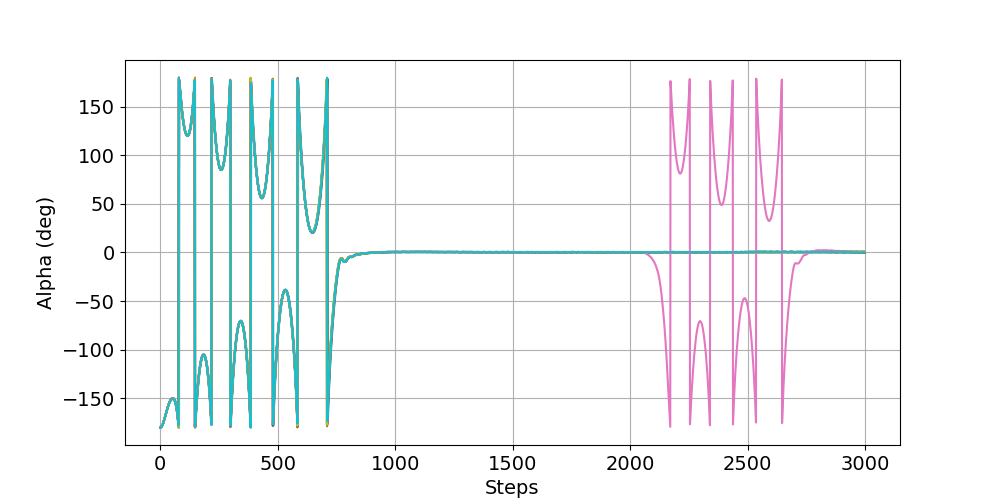}
    \includegraphics[width=0.48\linewidth]{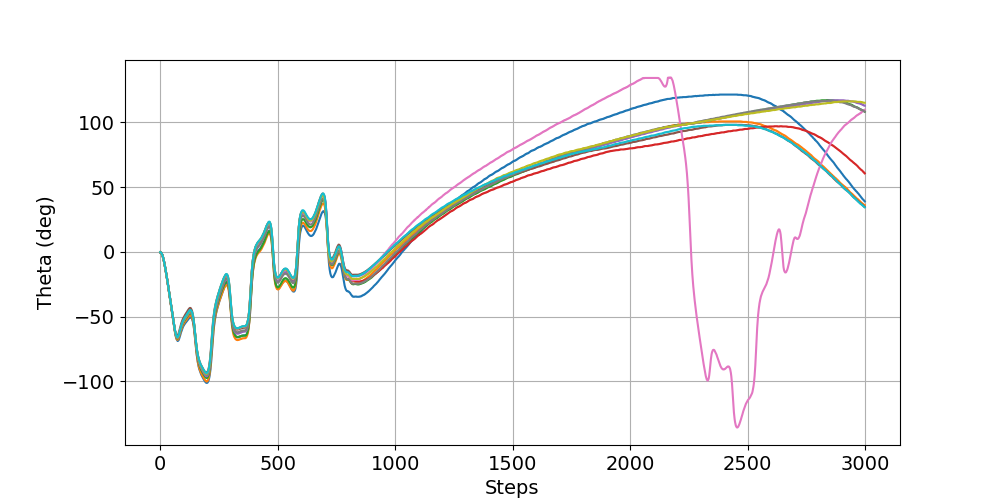}
    \caption{10 trials of CE stabilization. The pendulum almost achieves the upright posture (left) while fails to remain within $\pm90$ degrees (right). }
    \label{fig: quanser CE}
\end{figure*}

\begin{figure*}
    \centering
    \includegraphics[width=0.48\linewidth]{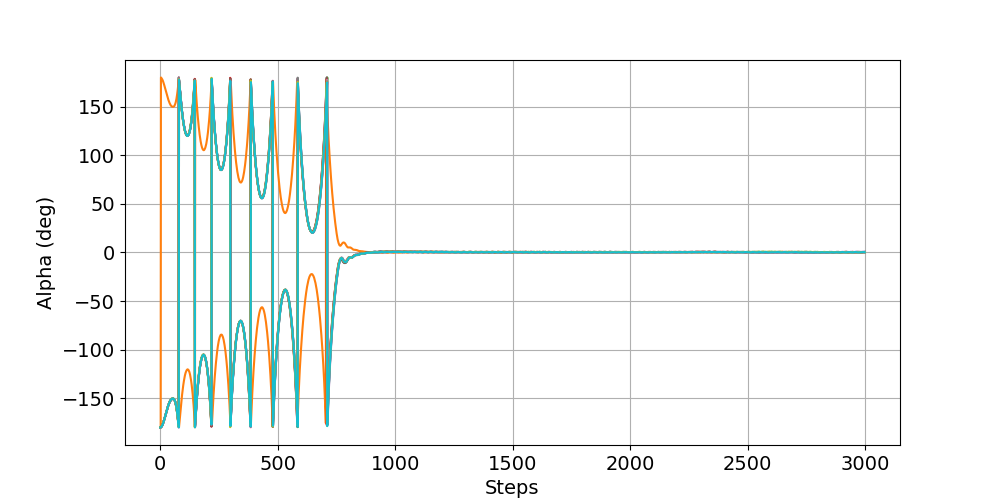}
    \includegraphics[width=0.48\linewidth]{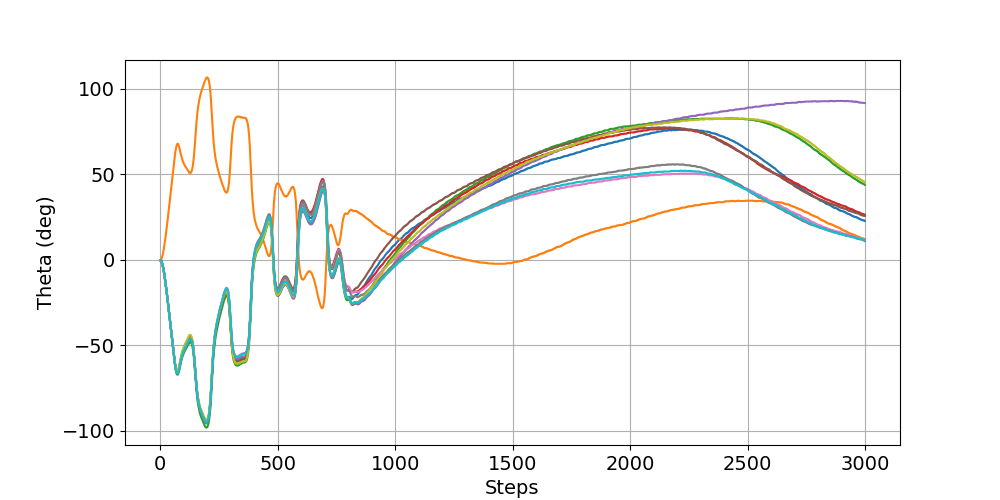}
    \caption{10 trials of DR stabilization. The pendulum almost achieves the upright posture (left) and remains withing $\pm90$ degrees simultaneously (right). }
    \label{fig: quanser DR}
\end{figure*}

The results are shown in \Cref{fig: quanser CE} and \Cref{fig: quanser DR}. The CE policy achieves $9$ \textbf{upright} and $1$ \textbf{failure} while the DR policy achieves $9$ \textbf{success }and $1$ \textbf{upright}. The reason the CE policy fails to hold the pendulum around the center can be explained by the sim-to-real gap. In particular, the optimal policy synthesized at the estimated parameter does not necessarily lead to stability when the parameter estimates are far from the ground truth. On the other hand, the DR policy succeeded in holding both upright and around center. This shows that it provides robustness to poor parameter estimates, thereby enabling sim-to-real transfer.


\end{document}